\documentclass[a4paper,USenglish,cleveref, autoref, thm-restate ]{lipics-v2021}
\nolinenumbers
 \usepackage[utf8]{inputenc}
 \usepackage{amsmath}     \usepackage{amssymb}     \usepackage{mathtools}   \usepackage{microtype}   \usepackage{multirow}    \usepackage{booktabs}    \usepackage{subcaption}  \usepackage[ruled,linesnumbered]{algorithm2e}  \usepackage[usenames,dvipsnames,table]{xcolor}  \usepackage{tikz}
\usetikzlibrary{shapes.geometric,arrows.meta,arrows}
\usetikzlibrary{calc}
\usepackage{nag}   \usepackage{tcolorbox}  \usepackage{floatrow}
\usepackage{hyperref}

\newcommand{\td}[0]{\mathsf{td}}
\newcommand{\vi}{\mathsf{vi}}
\newcommand{\con}{\textsf{\textup{Comp}}}

\newcommand{\cc}[1]{{\mbox{\textnormal{\textsf{#1}}}}\xspace} 
\newcommand{\NP}{\cc{NP}}
\newcommand{\FPT}{\cc{FPT}}
\newcommand{\XP}{\cc{XP}}

\newcommand{\sep}{\;|\;}
\newcommand{\seq}{\subseteq}

\newcommand{\bigoh}{\mathcal{O}}
\newcommand{\tww}{\textnormal{tww}}
\newcommand{\otww}{\textnormal{otww}}

\newcommand{\ext}{\textnormal{Ext}}

\newcommand{\seg}{\ensuremath{Y}}
 \newcommand{\C}{\mathcal{C}}
\newcommand{\blob}{large group\xspace}
\newcommand{\blobs}{large groups\xspace}

\usepackage{pgfplots}
\pgfplotsset{compat=1.15}
\usepackage{mathrsfs}
\usepackage{tikz}
\usetikzlibrary{arrows}
\usetikzlibrary{backgrounds}
\usetikzlibrary{hobby}
\usetikzlibrary{decorations.markings}
\usetikzlibrary{math}

\usepackage{nameref,cleveref}
\Crefname{splemma}{Lemma}{Lemmas}
\Crefname{sptheorem}{Theorem}{Theorems}
\Crefname{spdefinition}{Definition}{Definitions}
\Crefname{spproperty}{Property}{Properties}
\Crefname{spcorollary}{Corollary}{Corollaries}

\newtheorem{result}{Main Result}
\newtheorem{fact}[theorem]{Fact}

\newtheorem{redrule}{Reduction Rule}
\newcommand{\ie}{\emph{i.e.}\xspace}

\title{Computing Twin-Width \\ 
via Treedepth 
and Vertex Integrity}

\titlerunning{Approximating Twin-Width is FPT Parameterized by Treedepth}
 \Copyright{Robert Ganian and Mathis Rocton}
 
\author{Robert Ganian}{Algorithms and Complexity Group, TU Wien, Vienna, Austria}{rganian@gmail.com}{0000-0002-7762-8045}{Robert Ganian acknowledges support by the FWF and WWTF Science Funds (FWF projects 10.55776/Y1329 and 10.55776/COE12, WWTF project ICT22-029).}

\author{Mathis Rocton}{Algorithms and Complexity Group, TU Wien, Vienna, Austria}{mrocton@ac.tuwien.ac.at}{0000-0002-7158-9022}{Mathis Rocton acknowledges support by the \includegraphics[width=0.5cm]{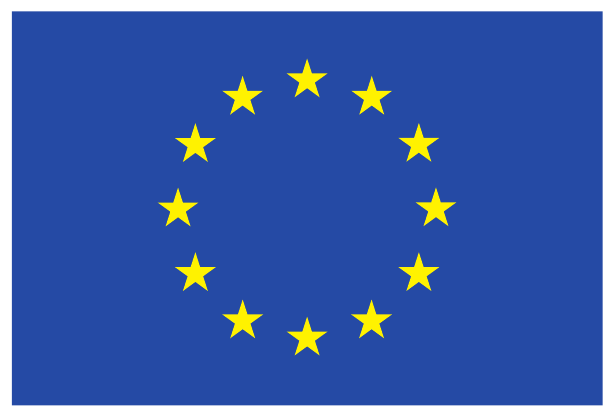} European Union's Horizon 2020 research and innovation COFUND programme (LogiCS@TUWien, grant agreement No 101034440), and the FWF Science Fund (FWF project Y1329).}

\authorrunning{R. Ganian and M. Rocton}


\begin{document}
\hideLIPIcs
\maketitle

\begin{abstract}
 Twin-width is a graph parameter that has become central to explaining the fixed-parameter tractability of first-order model checking across many graph classes. Despite its algorithmic importance, computing twin-width remains poorly understood: even recognizing graphs of twin-width at most four is NP-hard, and no fixed-parameter approximations parameterized by twin-width itself are known. A recent approach towards breaking this barrier focuses on first developing fixed-parameter algorithms for computing or approximating twin-width under parameterizations distinct from twin-width.

Our first result establishes that approximating twin-width is fixed-parameter tractable when parameterized by treedepth, thereby breaking the long-standing barrier that all previous tractable parameterizations were based on deletion distance. The proof proceeds via oriented twin-width, yielding the first constructive evidence that this variant may be easier to handle algorithmically. As our second main result, we show that computing twin-width exactly is fixed-parameter tractable with respect to vertex integrity. This constitutes the first non-trivial parameterized algorithm for computing optimal contraction sequences. 
\end{abstract}

\newpage

\section{Introduction}\label{sec:intro}
\emph{Twin-width} is a relatively recent graph-theoretic parameter that has both emerged from and stimulated a series of breakthroughs in algorithmic model theory~\cite{DBLP:conf/soda/BonnetGKTW21,DBLP:conf/icalp/BonnetG0TW21,DBLP:conf/stoc/BonnetGMSTT22,DBLP:conf/stacs/BonnetGMT23,DBLP:conf/soda/BonnetKRT22}. 
It holds the promise of offering a unifying explanation for why first-order logic model-checking is fixed-parameter tractable on a wide range of graph classes which, until recently, were viewed as isolated ``islands of tractability.'' 
Examples include graphs of bounded rank-width, proper minor-closed graphs, bounded-width posets~\cite{DBLP:conf/iwpec/BalabanH21}, map graphs~\cite{BonnetKTW22}  as well as a number of other specialized graph classes~\cite{DBLP:conf/wg/BalabanHJ22,EppsteinConfluent}.

Although twin-width is related to other parameters such as path-width, clique-width and rank-width~\cite{DBLP:conf/soda/BonnetKRT22}, it stands apart in one crucial respect: no efficient algorithms are known for computing it. In fact, already deciding whether a graph has twin-width $4$ is \NP-hard~\cite{BergeBD22}. This poses a serious challenge, as essentially all known algorithms that exploit twin-width require access to a \emph{contraction sequence}, which plays a role analogous to that of decompositions for classical parameters such as treewidth~\cite{RobertsonS83} and rank-width~\cite{Oum05}. 
Intuitively, a contraction sequence of width $t$---which exists whenever $G$ has twin-width at most $t$---is a sequence $\C$ of contractions of (not necessarily adjacent) vertex pairs such that, at each step of $\C$, every vertex $v$ has at most $t$ neighbors with an ancestor not adjacent to some ancestor of $v$; see Section~\ref{sec:prelims} for formal details. 

The \NP-hardness of recognizing graphs of twin-width at most $4$~\cite{BergeBD22} rules out both fixed-parameter and \XP\ algorithms for computing optimal contraction sequences when parameterized by twin-width itself. 
A natural workaround would be to design a fixed-parameter algorithm (parameterized by the twin-width $t$) that computes an \emph{approximate} contraction sequence, i.e., one of width $f(t)$ for a computable function $f$. From a complexity-theoretic standpoint, such a result would be nearly as powerful as computing the exact twin-width as it would still enable fixed-parameter tractability of first-order model checking. In fact, two decades ago it was the discovery of an analogous approximate decomposition-computing result~\cite{OumS06} which opened the door to the algorithmic applications of clique-width.

Unfortunately, finding such an algorithm for twin-width has proven to be highly challenging, and it remains unclear whether it is even possible. Indeed, determining whether twin-width admits any fixed-parameter approximation (for arbitrary computable $f$) is arguably the central open problem in current research on the parameter. 
In their recent work, Balab\' an, Ganian and Rocton~\cite{BalabanGR24} approached this question by relaxing the runtime requirements and asked whether one could obtain an $f(t)$-approximation for twin-width via a fixed-parameter algorithm parameterized by a graph measure other than twin-width itself. As a first step, they developed a fixed-parameter algorithm that computes a contraction sequence of width at most $t+1$, parameterized by the \emph{feedback edge number} of the graph~\cite{BalabanGR24}, i.e., by the edge deletion distance to a forest. In a follow-up work, the same authors obtained a fixed-parameter approximation algorithm which computes a contraction sequence of width at most $2t$ and is parameterized by the \emph{vertex integrity} of the graph~\cite{BalabanGR24IPEC}. Vertex integrity can be roughly understood as the deletion distance to a collection of small components, and the aforementioned $2$-approximation algorithm can be seen as a way of lifting the trivial fixed-parameter algorithm for computing optimal contraction sequences parameterized by the vertex cover number of the graph~\cite{BalabanGR24,BalabanGR24IPEC}.

\smallskip
\noindent \textbf{Results.}\quad
In this article, we push beyond the state of the art for computing twin-width in two directions. As our first result, we lift the fixed-parameter approximability of twin-width from the vertex integrity parameterization to treedepth ---a well-studied parameter that has fundamental ties to the theory of graph sparsity~\cite{sparsity}. In particular, we prove:

\begin{result}
\label{result1}
Approximating twin-width is \FPT\ w.r.t.\ the treedepth of the input graph.
\end{result}

Even though the algorithm underlying Main Result~\ref{result1} (formalized in Theorem~\ref{thm:mainone}) provides approximation guarantees only in terms of a superexponential function of the twin-width, we believe it to be a major step forward in our conceptual understanding. Indeed, while all previous parameterizations which supported efficient approximation algorithms for twin-width were based on deletion distance, twin-width itself---as well as major structural parameters such as treewidth and clique-width---are \emph{decompositional} in nature. Our result thus marks the first time we are able to break the barrier towards decompositional parameters: having bounded treedepth is typically witnessed via a certain type of decomposition (which resembles a bounded-depth tree), but cannot be expressed via a constant number of deletion operations\footnote{We remark that both treedepth and treewidth admit characterizations that rely on vertex deletion operations---however, in both cases the number of such operations is not bounded by the parameter.}.

A further noteworthy property of Theorem~\ref{thm:mainone} is that the proof does not target the computation of twin-width directly. Instead, it relies on reduction rules which maintain a $2$-approximate bound on the so-called \emph{oriented twin-width}---a variant of twin-width where contraction sequences only capture information about edges of the original graph in a one-sided way. Oriented twin-width has been proposed and studied already in the pioneering series of works that introduced twin-width~\cite{DBLP:conf/soda/BonnetKRT22}; it was known to be functionally equivalent to twin-width~\cite[Theorem 4.1]{DBLP:conf/soda/BonnetKRT22}, and in Section~\ref{sec:translation} we make this relationship both explicit and constructive. It was suspected that oriented twin-width could be easier to compute, and Theorem~\ref{thm:mainone} yields concrete substance to that claim: the algorithm achieves a $2$-approximation on the oriented twin-width of the input graph, but it is entirely unclear how to obtain any such bound for ``vanilla'' twin-width.

While our first result relies on allowing for a potentially large approximative error in the twin-width, the second one makes do without this entirely: 

\begin{result}
\label{result2}
Computing twin-width is \FPT\ w.r.t.\ the vertex integrity of the input graph.
\end{result}

Main Result~\ref{result2} (formalized in Theorem~\ref{thm:maintwo}) is the first non-trivial parameterized algorithm for computing optimal contraction sequences, and shows that the \NP-hardness of the problem can be overcome at least under restrictive parameterizations. Moreover, while its running time guarantees are weak in practice, the underlying algorithm relies on the performance of reduction rules which are easy to implement and guaranteed to preserve the twin-width of the input graph. We remark that both algorithms obtained in this work are deterministic, constructive and rely exclusively on computable functions.

\smallskip
\noindent \textbf{Technical Contributions.}\quad
On a high level, the proof of Theorem~\ref{thm:mainone} relies on establishing that every sufficiently large graph of bounded treedepth must have a ``redundant'' subgraph $H$---that is, a subgraph which we can find and safely delete from the instance. Recursively applying such a rule will eventually reduce the graph to a problem kernel~\cite{CyganFKLMPPS15}, 
 at which point one may compute a contraction sequence via brute force or any other known algorithm. 
It is well-known that deleting vertices cannot increase the twin-width of a graph, but for correctness it is also necessary to prove the converse: that by deleting $H$, we have not accidentally reduced the twin-width of $G$. In other words, a redundant subgraph must have the property that it can be reinserted into the graph without increasing the twin-width. 

The main obstacle towards approximating twin-width when parameterized by treedepth is that---unlike many problems where the above recursive deletion technique has been successfully applied---reinserting $H$ may require the contraction sequence to have a short subsequence where the twin-width is elevated. In particular, the \emph{red degree} (see Section~\ref{sec:prelims}) can exceed the twin-width bound by a factor of at most $2$, and this can happen both for vertices in $H$ and vertices outside of $H$. For a single reinsertion step this would not be an issue, but the algorithm must be able to perform deletion multiple times, which would then lead to the error gradually snowballing to a factor of $4$, $8$, and so forth. The crucial insight here is that if we instead consider contraction sequences for oriented twin-width, the red degree exceeds the twin-width bound only for vertices inside $H$; in turn, this allows us to isolate the approximation error in clearly delimited and pairwise disjoint parts of the contraction sequence.

For Theorem~\ref{thm:maintwo}, we build on the very recently introduced Ramsey Pruning technique~\cite{DFGS25}. The core idea behind that technique is that instead of repeatedly arguing that  a single redundant subgraph can be safely reinserted if it is deleted, we use Ramsey-type arguments to argue that every contraction sequence $S$ for $G$ must contain a carefully selected subgraph $Q$ which induces ``guiding subsequence'' $S'$---a contraction sequence for $Q$ that is highly structured in a precisely defined way. Once defined, the properties of $S'$ will allow us to argue that every solution on $Q$ can be lifted to a contraction sequence of width $t$ for an infinite set of supergraphs of $Q$ which also includes $G$. Thus, while we do not directly prove the safeness of performing any single deletion operation, the proof technique guarantees that the existence of a solution (in this case, a contraction sequence of width $t$) is preserved between the first and final graph in the sequence of reduction rules.

The main difficulty that arises when trying to implement the above strategy in the setting of twin-width is that, essentially, the properties of $S'$ one can guarantee by invoking the analogous Ramsey-type arguments as in the preceding work~\cite{DFGS25} are too weak. There, the core idea was to capture pairwise relationships between certain subgraphs of $G$ via color-coded edges in an auxiliary graph, and then employ Ramsey theory to guarantee the existence of a monochromatic subclique. Unfortunately, pairwise relationships do not capture sufficient information to guarantee the ability to expand $S'$ into a solution for $G$ (or a supergraph thereof). We solve this by showing that maintaining information about interactions between triplets of certain subgraphs in $G$ does suffice---a complication which necessitates the use of Ramsey hypergraph theory instead of the classical bounds used in~\cite{DFGS25}.

\subparagraph*{Related Work.}
Beyond the computation of twin-width and its associated contraction sequences, a substantial body of work has focused on fixed-parameter algorithms for computing a structural graph parameter $A$ when parameterized by graph parameters different from $A$. The overarching goal of this research direction is to deepen our understanding of the fundamental task of computing the target parameter $A$. Prominent examples include fixed-parameter algorithms for treedepth parameterized by the vertex cover number~\cite{DBLP:conf/iwpec/KobayashiT16}, treewidth parameterized by the feedback vertex number~\cite{DBLP:journals/siamdm/BodlaenderJK13}, MIM-width parameterized by the feedback edge number and other parameters~\cite{DBLP:conf/innovations/EibenGHJK22}, as well as directed feedback vertex number parameterized by the undirected feedback vertex number~\cite{DBLP:journals/algorithmica/BergougnouxEGOR21}.

Treedepth and vertex integrity have also served as effective parameterizations for developing algorithms for a variety of hard problems~\cite{BannisterCE18,GanianO18,LampisM21,BhoreGMN22,GimaHKKO22,FichteGHSO23,GimaO24}. Moreover, vertex integrity has been studied under several asymptotically equivalent notions in the literature, including the \emph{fracture number}~\cite{DvorakEGKO21,GanianKO21} and \emph{starwidth}~\cite{Ee17}.

\section{Preliminaries}
\label{sec:prelims}

For integers $i$ and $j$, we let $[i,j] := \{n \in \mathbb N \sep i \le n \le j\}$ and $[i] := [1, i]$. 
We assume familiarity with basic concepts in graph theory~\cite{Diestel} and parameterized algorithmics~\cite{DowneyF13,CyganFKLMPPS15}.
When $H$ is an induced subgraph of $G$, we denote it by $H \seq G$.
Given vertex sets $X$ and $U$, we will use $G[X]$ to denote the graph induced on $X$ and $G - U$ to denote the graph $G[V(G)\setminus U]$. 
A vertex $v$ of a rooted tree $T$ is a \emph{descendant} of a vertex $w$ if $w$ occurs on the root-to-$v$ path in $T$; we then call $w$ an \emph{ancestor} of $v$.
We recall that the number of non-isomorphic graphs of size up to $n$ can be upper-bounded by $n\cdot 2^{n^2}$.
To express some of our bounds, we will occasionally use the Knuth notation $\uparrow\uparrow$ where for an integer $z$, $2\uparrow\uparrow z$ represents an exponential tower of $2$'s of height $z$.

\smallskip
\noindent \textbf{Treedepth.} 
A \emph{treedepth decomposition} of a graph $G$ is a pair $(T,f)$, where $T$ is a rooted forest and $f : V(G) \rightarrow V(T)$ is a bijective function such that for each $\{u,v\} \in E(G)$, either $f(u)$ is a descendant of $f(v)$ or $f(v)$ is a descendant of $f(u)$. The \emph{depth} of the treedepth decomposition is the number of vertices in the longest root-to-leaf path in $T$; for ease of presentation, w.l.o.g.\ we will assume that $f$ is just an identity and simply write $T$ for $(T, f)$. The \emph{treedepth} of a graph $G$, denoted by $\td(G)$, is the minimum over the depths of all possible treedepth decompositions of $G$. When $G$ is connected, $T$ is a tree. For any vertex $u$ of $T$, we denote by $T_u$ the subtree of $T$ rooted at $u$, and $X_u$ the subgraph $G[V(T_u)]$.
 
A treedepth decomposition $T$ is \emph{nice} if, for each node $u\in V(T)$ and each child $w$ of $u$ in $T$,  $X_w$ forms a connected component of $X_u \setminus\{u\}$---in particular, there is a one-to-one correspondence between the subtrees rooted at the children of $u$ in  $T$ and the set $\con_u$ of connected components in $X_u \setminus \{u\}$.
 It is known that every treedepth decomposition can be transformed into a nice one with a smaller or equal depth via a simple rearrangement argument~\cite{ReidlRVS14}, and that a nice treedepth decomposition of depth $\td(G)$ can be computed in time $2^{\bigoh(\td(G)^2)}\cdot |V(G)|$~\cite{ReidlRVS14}, see also the recent work~\cite{NadaraPS22} on the topic.

For ease of presentation, we say that a subgraph $H$ of $G$ \emph{appears in $T$} if there exists a \emph{seed} vertex $u$ in $T$ such that $H=X_u$.
  Moreover, we call two subgraphs $X_a$, $X_b$ appearing in $T$ \emph{siblings} if $a$ and $b$ are siblings in $T$, i.e., have the same parent.

\smallskip
\noindent \textbf{Vertex Integrity.} A graph $G$ has \emph{vertex integrity} $\vi(G) = p$ if $p$ is the smallest integer with the following property: $G$ contains a vertex set $S$ such that for each connected component $H$ of $G - S$, $|V(H) \cup S| \leq p$. One may observe that the vertex integrity is upper-bounded by the size of a minimum vertex cover in the graph (i.e., the vertex cover number) plus one. The vertex integrity of an $n$-vertex graph along with a corresponding partition into $S$ and $\mathcal{C}=G-S$ can be computed in time $\mathcal{O}(p^{p+1} \cdot n)$~\cite{DrangeDH16}.

\smallskip
\noindent \textbf{Hypergraph Ramsey Bounds.}\quad
We will employ a known generalization of Ramsey's classical theorem to edge-colored hypergraphs in order to argue the existence of certain structures in our correctness proof. 
In particular, the following fact follows directly from the application of Erdős' and Rado's improved bound on the Ramsey numbers for multicolored hypergraphs~\cite[Theorem 1]{erdos1952combinatorial}; 
 see also the recent work on the topic~\cite{bradavc2024growth} for an overview.
 
\begin{fact}
\label{fact:Ramsey}
Let $H$ be a complete hypergraph whose hyperedges all have size $3$ and are mapped to single colors from $[L]$, and let $t$ be a positive integer. If $|V(H)|\geq L^{L^{2L\cdot t}}$, then $H$ contains an induced $t$-vertex subhypergraph $H'$ such that all hyperedges fully contained in $H'$ have the same color.
 \end{fact}

\smallskip
\noindent \textbf{Twin-Width.}\quad
            Our notation for twin-width follows the conventions introduced in the recent related works aimed at computing twin-width under more restrictive parameterizations~\cite{BalabanGR24,BalabanGR24IPEC}.
A \emph{trigraph} $G$ is a graph whose edge set is partitioned into a set of \emph{black} and \emph{red} edges. The set of red edges is denoted $R(G)$, and the set of black edges $E(G)$.
The \emph{black (\emph{resp}.\ red) degree} of $u\in V(G)$ is the number of black (resp.\ red) edges incident to $u$ in $G$.
We extend graph-theoretic terminology to trigraphs by ignoring the colors of edges; for example, the degree of $u$ in $G$ is the sum of its black and red degrees.

Given a trigraph $G$, a \emph{contraction} of two distinct vertices $u,v\in V(G)$ is the operation which produces a new trigraph by (1) removing $u, v$ and adding a new vertex $w$, (2) adding a black edge $wx$ for each $x\in V(G)$ such that $xu$, $xv\in E(G)$, and (3) adding a red edge $wy$ for each $y\in V(G)$ such that $yu\in R(G)$, or $yv\in R(G)$, or $y$ has a black edge to either $v$ or $u$ (but not both).
For an integer $p$, a sequence $\C = (G = G_1,\ldots,G_p)$ is a \emph{partial contraction sequence of $G$} if it is a sequence of trigraphs such that for all $i\in [p-1]$, $G_{i+1}$ is obtained from $G_i$ by contracting two vertices; if $p=|V(G)|$ then $G_p$ is the single-vertex graph and we call $\C$ a \emph{contraction sequence}.
We call a contiguous subsequence of $\C$ a \emph{segment}.
The \emph{width} of a (partial) contraction sequence $\C$ is the maximum red degree over all vertices in all trigraphs in $\C$.  The \emph{twin-width} of $G$, denoted $\tww(G)$, is the minimum width of any contraction sequence of $G$, and a contraction sequence of width $\tww(G)$ is called \emph{optimal}. An example of a contraction sequence is provided in Figure~\ref{fig:seq}.

\begin{fact}[e.g., Observation 6 in~\cite{BalabanGR24}]
\label{fact:computecontract}
An optimal contraction sequence of an $n$-vertex graph can be computed in time $2^{\bigoh(n\cdot \log n)}$.
\end{fact}

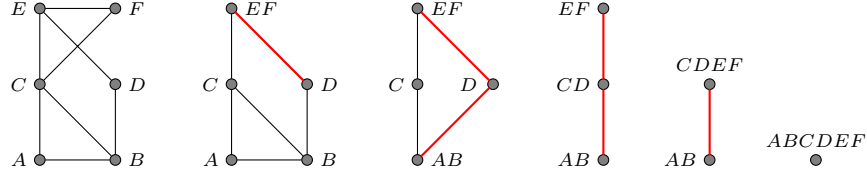
\begin{figure}[ht]
\begin{tikzpicture}[line cap=round,line join=round,>=triangle 45,x=1.0cm,y=1.0cm]

\tikzset{
    vertex/.style = {draw, circle, fill=gray, minimum width=4pt, inner sep=0pt}}
    
\begin{scriptsize}
\node[vertex] (a) {};
\node[vertex] (b) [right of = a] {};
\node[vertex] (c) [above of = a] {};
\node[vertex] (d) [right of = c] {};
\node[vertex] (e) [above of = c] {};
\node[vertex] (f) [right of = e] {};
\draw (a)--(b)--(c)--(f)--(e)--(d)--(b);
\draw(a)--(c)--(e);

\node () at (a) [left=2pt] {$A$};
\node () at (c) [left=2pt] {$C$};
\node () at (e) [left=2pt] {$E$};
\node () at (b) [right=2pt] {$B$};
\node () at (d) [right=2pt] {$D$};
\node () at (f) [right=2pt] {$F$};

\node[vertex] (a2) [right = 70pt]{};
\node[vertex] (b2) [right of = a2] {};
\node[vertex] (c2) [above of = a2] {};
\node[vertex] (d2) [right of = c2] {};
\node[vertex] (ef) [above of = c2] {};
\draw (a2)--(b2)--(c2)--(ef)--(d2)--(b2);
\draw(a2)--(c2);
\draw[color=red, thick] (d2)--(ef);

\node () at (a2) [left=2pt] {$A$};
\node () at (c2) [left=2pt] {$C$};
\node () at (b2) [right=2pt] {$B$};
\node () at (d2) [right=2pt] {$D$};
\node () at (ef) [right=2pt] {$EF$};

\node[vertex] (ab) [right = 140pt]{};
\node[vertex] (c3) [above of = ab] {};
\node[vertex] (d3) [right of = c3] {};
\node[vertex] (ef3) [above of = c3] {};
\draw (ab)--(c3)--(ef3);
\draw[color=red, thick] (ab)--(d3)--(ef3);

\node () at (ab) [right=2pt] {$AB$};
\node () at (c3) [left=2pt] {$C$};
\node () at (d3) [left=3pt] {$D$};
\node () at (ef3) [right=2pt] {$EF$};

\node[vertex] (ab4) [right = 210pt] {};
\node[vertex] (cd) [above of = ab4]{};
\node[vertex] (ef4) [above of = cd] {};
\draw[color=red, thick] (ab4)--(cd)--(ef4);

\node () at (ab4) [left=2pt] {$AB$};
\node () at (cd) [left=2pt] {$CD$};
\node () at (ef4) [left=2pt] {$EF$};

\node[vertex] (ab5) [right = 250pt] {};
\node[vertex] (cdef) [above of= ab5]{};
\draw[color=red, thick] (ab5)--(cdef);
\node () at (ab5) [left=2pt] {$AB$};
\node () at (cdef) [above = 2pt] {$CDEF$};

\node[vertex] (fin) [right = 290pt] {};
\node () at (fin) [above = 2pt] {$ABCDEF$};

\end{scriptsize}
\end{tikzpicture}
\caption{A contraction sequence of width 2 for the leftmost graph, consisting of $6$ trigraphs.
\label{fig:seq}}
\end{figure}

Let us now fix a contraction sequence $\C = (G = G_1,\ldots,G_n)$.
For each $i \in [n]$, we associate each vertex $u \in V(G_i)$ with a set $\beta(u, i) \seq V(G)$, called the \emph{bag} of $u$, which contains all vertices contracted into $u$.
Formally, we define the bags as follows:

\begin{itemize}
\item for each $u \in V(G)$, $\beta(u, 1) := \{u\}$;
\item for $i \in [n-1]$, if $w$ is the new vertex in $G_{i+1}$ obtained by contracting $u$ and $v$, then $\beta(w, i+1) := \beta(u, i) \cup \beta(v,i)$; otherwise, $\beta(w, i+1) := \beta(w, i)$. 
\end{itemize}

Note that if a vertex $u$ appears in multiple trigraphs in $\C$, then its bag is the same in all of them, and so we may denote the bag of $u$ simply by $\beta(u)$.
Let us fix $i, j \in [n]$, $i \le j$.
If $u \in V(G_i)$, $v \in V(G_j)$, and $\beta(u) \seq \beta(v)$, then we say that $u$ is \emph{a $\C$-ancestor} of $v$ in $G_i$ and $v$ is \emph{the $\C$-descendant} of $u$ in $G_j$ (clearly, the $\C$-descendant is unique).
If $H$ is an induced subtrigraph of $G_i$, then $u \in V(G_j)$ is a $\C$-\emph{descendant} of $H$ if it is a $\C$-descendant of at least one vertex of $H$.
  A contraction of $u, v \in V(G_j)$ into $uv \in V(G_{j+1})$ \emph{involves} $w \in V(G_i)$ if $w$ is a $\C$-ancestor of $uv$.
Given a contraction sequence $\C$ of $G$ and vertex subset $Z\subseteq V(G)$, we let $\C[Z]$ denote the \emph{restriction} of $\C$ to $Z$, i.e., the contraction sequence obtained from $\C$ by omitting each step that involves vertices outside of $Z$; note that $\C[Z]$ is a contraction sequence of $G[Z]$ whose width is upper-bounded by the width of $\C$.

The twin-width of a graph with vertex integrity $p$ is upper-bounded by $p$. Indeed, given a vertex set $S$ such that for each connected component $H$ of $G - S$, $|V(H) \cup S| \leq p$, we can construct a contraction sequence of width $p$ by processing components in an arbitrary order as follows. We contract the first component $H_1$ into a single vertex  $v_1$ (during which the red degree never exceeds $p$). Afterwards, for each unprocessed component $i$, $i\geq 2$, we contract $H_i$ into a single vertex (during which the red degree never exceeds $p$ either) and contract this vertex with $v_{i-1}$ to produce $v_i$.

\smallskip
\noindent \textbf{Oriented Twin-Width.}\quad
The oriented variant of twin-width is defined using the same concept of contraction sequences as twin-width, and hence we opt to use primarily shared terminology (as opposed to the homogeneity-based notation of~\cite{DBLP:conf/soda/BonnetKRT22}). Intuitively, in oriented twin-width red edges only represent discrepancies between the corresponding bags in a one-directional manner; this will require the individual trigraphs in a contraction sequence to depend not only on the previous trigraph in the sequence but also on the original graph. We formalize below.

An \emph{oriented trigraph} is a graph whose edge set is partitioned into a set of black undirected edges and red directed edges (arcs). Given an initial graph $G$ and an oriented trigraph $G_{i-1}$ each of whose vertices $z\in V(G_{i-1})$ is associated with a bag $\beta(z, i-1)\subseteq V(G)$, an \emph{oriented contraction} of two distinct vertices $u,v\in V(G_{i-1})$ 
is the operation which produces a new oriented trigraph $G_{i}$ as follows. First, we remove $u, v$ and add a new vertex $w$ where $\beta(w,i)=\beta(u,i-1)\cup \beta(v,i-1)$; for all other vertices, the bags remain the same and so do the edges and arcs not incident to $w$. For each $a\in V(G_i)$ we add a black edge $aw$ if every vertex of $\beta(a,i)$ is adjacent to every vertex of $\beta(w,i)$ in $G$. For each $b\in V(G_i)$ we add a red arc $(w,b)$ if there exist $w_1,w_2\in \beta(w,i)$ and $b_0\in \beta(b,i)$ such that $w_1b_0\in E(G)$ and $w_2b_0\not \in E(G)$. Similarly, for each $c\in V(G_i)$ we add a red arc $(c,w)$ if there exist $c_1,c_2\in \beta(c,i)$ and $w_0\in \beta(w,i)$ such that $c_1w_0\in E(G)$ and $c_2w_0\not \in E(G)$. Notice that an oriented contraction cannot increase the red out-degree of any vertex in $V(G_i)\setminus \{w\}$.

An \emph{oriented contraction sequence} is then defined analogously as a contraction sequence, but using oriented contractions. The \emph{oriented width} of an oriented contraction sequence $\C$ the 
maximum red out-degree over all vertices in all trigraphs in $\C$. While the oriented twin-width of a graph can be smaller than its twin-width, by comparing the definitions of the two notions of contraction sequences we immediately obtain:

\begin{fact}[\cite{DBLP:conf/soda/BonnetKRT22}]
\label{fact:otwwtwwbound}
The oriented twin-width of every oriented contraction sequence $C$ is upper-bounded by the width of the contraction sequence $C'$ obtained by the same sequence of vertex pair contractions.
\end{fact}

The \emph{oriented twin-width} of a graph $G$, denoted $\otww(G)$, is the minimum oriented width of any oriented contraction sequence of $G$.

\begin{remark*}
Since Section~\ref{sec:td} will deal exclusively with oriented contraction sequences and oriented width, for brevity we omit the adjective ``oriented'' in that section.\end{remark*}

\section{Translating Twin-Width and Oriented Twin-Width}
\label{sec:translation}

Bonnet, Kim, Reinald and Thomass\' e established the functional equivalence between twin-width and its oriented variant by showing:

\begin{fact}[Proof of Theorem 4.1 in~\cite{DBLP:conf/soda/BonnetKRT22}]
For every graph $G$, $\otww(G)\leq \tww(G) \leq 2^{2^{\bigoh(\otww(G)}}$.
\end{fact}

Unfortunately, the proof of the above statement is not constructive---it does not provide any algorithm for converting an oriented contraction sequence $\C$ of oriented width $\alpha$ into a contraction sequence $\C'$ of width bounded by a function of $\alpha$. In order to construct the contraction sequences for Main Result~\ref{result1}, we need to revisit the relationship between the two parameters and provide an efficient conversion algorithm.

First, we provide some background for the terminology required solely in this section.
Given a total order $\sigma$ over $V(G)$, the \emph{ordered adjacency matrix} $M(G,\sigma)$ is the adjacency matrix of $G$ where the rows and columns both appear in the order $\sigma$. The mixed number is a numerical parameter of ordered adjacency matrices which is known to be related to twin-width, but we will not need its definition for our purposes. 
Similarly, a notion of \emph{twin-width} has also been defined over ordered adjacency matrices but we will not need its definition for our purposes. 
We say that a (oriented) contraction sequence in $C$ of $G$ \emph{respects} $\sigma$ if the vertices in every bag in every trigraph of $C$ occur consecutively in $\sigma$---in other words, for each pair of vertices $a,b$ occurring in some bag $\beta(\cdot, \cdot)$, there cannot exist any $c\not \in \beta(\cdot, \cdot)$ such that $c$ occurs between $a$ and $b$ in $\sigma$.
     
    \begin{lemma}
 \label{lem:otwwmixed}
Given a graph $G$ with an oriented contraction sequence $\C$ of oriented width at most $k$, we can compute in polynomial time an ordering $\sigma$ on $V(G)$ such that $(M(G,\sigma))$ has mixed number at most $2k+3$.
  \end{lemma}

\begin{proof}
Let us fix an arbitrary total order $\prec$ over $V(G)$. From $\C=(G = G_1,\ldots,G_n)$, we build a sequence of total orders $\sigma_n,\dots, \sigma_1$ over $V(G_n), \dots, V(G_1)$, respectively, as follows. $\sigma_n$ is the trivial unique total order over a single element. For each $i\in [n-1]$, we then construct $\sigma_{i}$ from $\sigma_{i+1}$ by maintaining the same pairwise ordering between each pair of bags not involved in the $i$-th contraction of $\C$. Let $\beta{v,i+1}$ be the unique bag in $G_{i+1}$ that was created by contracting two bags in $G_i$, say $\beta(a, i)$ and $\beta(b, i)$ where $\min_\prec(\beta(a, i))\prec \min_\prec(\beta(b,i))$. We complete the construction of $\sigma_i$ by replacing $\beta(v,i+1)$ in the ordering $\sigma_{i+1}$ with $\beta(a, i)$ directly followed by $\beta(b, i)$.
Set $\sigma=\sigma_1$ to be the resulting ordering on $V(G)$; intuitively, $\sigma$ is the natural ordering of the vertices on the lowest level of the so-called contraction tree of $\C$~\cite{DBLP:conf/soda/BonnetGKTW21}, with $\prec$ simply serving as a means of making the construction deterministic. Notice that the construction guarantees that $\C$ respects $\sigma$.

Let us consider the ordered adjacency matrix $M(G,\sigma)$, and towards a contradiction assume that $(M(G,\sigma))$ has mixed number at least $2k+4$. By a verbatim reproduction of the last two paragraphs of the proof of~\cite[Theorem 4.1]{DBLP:conf/soda/BonnetKRT22} where we set $d\coloneq k+1$, we obtain that every oriented contraction sequence which respects $\sigma$ must have oriented width at least $k+1$. 
\end{proof}

The first part of the following statement follows directly from the second part of the Grid Minor Theorem for Twin-Width~\cite[Theorem 10]{BonnetKTW22}, see also~\cite[Theorem 2.2]{DBLP:conf/soda/BonnetGKTW21}, when setting $t\coloneq 2k+3$. The second part follows from~\cite[Theorem 5.8]{BonnetKTW22}, which translates the aforementioned grid minor theorem to graphs. We note that the statement of~\cite[Theorem 5.8]{BonnetKTW22} omits the fact that the constructed contraction sequence respects $\sigma$, but this follows directly from the proof.

\begin{fact}
\label{fact:mixedtww}
Every ordered matrix $(M(G),\sigma)$ of mixed number $2k+3$ has twin-width at most $2^{2^{\bigoh(k)}}$. Moreover, there exists a contraction sequence of $G$ which respects $\sigma$ and has width at most $2^{2^{\bigoh(k)}}$.
\end{fact}

We remark that $100\cdot 2^{2^{12k}}$ can serve as a non-tight but explicit bound for the above terms involving $k$. Next, we recall the fixed-parameter tractability of approximating contraction sequences with respect to a provided total order of the vertex set~\cite[Theorem 7]{BonnetGMSTT24}.

\begin{fact}
\label{fact:computerespecting}
There is a fixed-parameter algorithm that takes as input a graph $G$ with a total order $\sigma$ of its vertices, and outputs a contraction sequence of $G$ which respects $\sigma$ of width at most $2^{\bigoh(p^4)}$. Here, the parameter $p$ is the minimum width among all contraction sequences of $G$ which respect $\sigma$. 
\end{fact}

By chaining the above arguments, we obtain:

 \begin{lemma}
\label{lem:translation}
There is a fixed-parameter algorithm that takes as input a graph $G$ together with a contraction sequence $\C$ of oriented width $k$ (the parameter) and outputs a contraction sequence $\C'$ of width at most $2^{2^{2^{\bigoh(k)}}}$. Moreover, if $k\leq f(\otww(G))$ for some computable non-decreasing function $f$, then $\C'$ has width at most $2^{2^{2^{\bigoh(f(\tww(G)))}}}$.
 \end{lemma}

\begin{proof}
Using Lemma~\ref{lem:otwwmixed}, we compute in polynomial time an ordering $\sigma$ on $V(G)$ such that $(M(G,\sigma))$ has mixed number at most $2k+3$. Fact~\ref{fact:mixedtww} then stipulates the existence of a contraction sequence $C^*$ of $G$ which respects $\sigma$ and has width at most $2^{2^{\bigoh(k)}}$. Next, we invoke the fixed-parameter algorithm from Fact~\ref{fact:computerespecting} on $(G,\sigma)$. Here, the existence of $C^*$ guarantees that the contraction sequence $C'$  computed by the algorithm will have width at most $2^{(2^{2^{\bigoh(k)}})^4}\in 2^{2^{2^{\bigoh(k)}}}$.

For the second statement, it suffices to recall that $\otww(G)\leq \tww(G)$ by Fact~\ref{fact:otwwtwwbound}.
\end{proof}

\section{A Fixed-Parameter Algorithm Parameterized by Treedepth}\label{sec:td}

In this section, we design an FPT 2-approximation algorithm for computing oriented twin-width when parameterized by the treedepth, see Theorem~\ref{thm:param-by-TD}. In combination with Lemma~\ref{lem:translation},
 this will in turn imply Main Result~\ref{result1} as formalized in Theorem~\ref{thm:mainone}.
 
\subsection{Initial Setup and Overview}
For the following, it will be useful to recall the definition of treedepth presented in Section~\ref{sec:prelims}.

Let $G$ be an arbitrary graph with treedepth $p$ and $T$ a nice treedepth decomposition of $G$ of depth $p$. Note that $G$ and $T$ have the same set of vertices.
We assume without loss of generality that the input graph $G$ is connected, as the oriented twin-width of a graph is the maximum oriented twin-width of its connected components; connectivity will then be preserved throughout our reduction rules.
For a given vertex $u$ in $T$, we now define a notion of ``component-types'' which intuitively captures the equivalence between components which exhibit the same outside connections and internal structure. 

\begin{definition}\label{def:equivalence}
We say that two graphs $H_0, H_1 \in \con_u$ appearing in $T$ are \emph{$u$-twin-blocks}, denoted $H_0 \sim_u H_1$, if there exists a \emph{canonical isomorphism} $\alpha$ from $H_0$ to $H_1$ such that for each vertex $v\in V(H_0)$ and each $w\in V(G\setminus (H_0+H_1))$, $vw\in E(G)$ if and only if $\alpha(v)w \in E(G)$.  Clearly, $\sim_u$ is an equivalence relation.
 \end{definition}

When $u$ is clear from context, we simply refer to $H_0$ and $H_1$ as \emph{twin-blocks} ($H_0\sim H_1$), and we use the standard notation $[H]_\sim$ to denote the equivalence class of $H$ for $\sim$. We note that $\sim$ can be tested in time $\bigoh(|H_0|^{|H_0|}\cdot p)$ via isomorphism testing procedures. Moreover, we lift $\alpha$ to sets of vertices by simply setting $\alpha(S)=\{\alpha(s)~|~s\in S\}$.

The core of our approach lies in establishing a reduction rule which takes a graph $G$ with its treedepth decomposition $T$ and gradually deletes subgraphs from $G$ until we obtain a subgraph whose size is upper-bounded by a function of $\td(G)$.
 
\begin{definition}\label{def:branchingfactor}
A vertex $u$ at depth $d$ in $T$ is \emph{processed} if: $T_u\setminus \{u\}$ contains only processed vertices, and $|\con_u| \leq g(d, p)$, where we set $g(p, p)=2^p$ and, for any $i \in [0, p-1]$, $g(i, p)\coloneq 2^{g(i+1, p)^{6p}}$.
\end{definition}

\begin{lemma}
\label{lem:size_processed}
If a subtree $T_u$ of $T$ is rooted at depth $d$ and $u$ is processed, it contains at most $b(d,p)\coloneq(g(d,p)+1)^p$ vertices.
In particular, if the root of $T$ is processed, $T$ contains at most $(g(1,p)+1)^p=2\uparrow\uparrow \bigoh(p)$ vertices.
\end{lemma}

\begin{proof}
If the root of $T_u$ is processed, all of its vertices are processed by definition. Thus, combining the bound from \cref{def:branchingfactor} at each level of $T$, knowing that the maximal depth is $p$, we obtain: $\Pi_{\delta \in \{d, p\}} (g(\delta,p)+1) \leq (g(d,p)+1)^p$. Indeed, notice that for any fixed $p\in \mathbb{N}$, the function $g(.,p)$ is decreasing.  
The result for $T$ follows directly from the application of the formula to $\delta=1$, and by \cref{def:branchingfactor} we can construct recursively the tower $2^{\dots^{2^{6p^2}}}$ with exactly $p$ stacked exponentials, which can be upper-bounded by $2\uparrow\uparrow \bigoh(p)$.
\end{proof}

\begin{definition}\label{def:mergeseg}
Let $G$ be a graph, $T$ a treedepth decomposition of $G$, $v\in V(G)$ and $\C$ a contraction sequence for $G$. We say that a segment $\seg$ of $\C$ is an \emph{$X_v$-merge} if $v$ has a sibling $w$ in $T$ with the following properties:
  \begin{itemize}
\item $X_v\sim X_w$, i.e., they are twin-blocks in $G$.
\item In the first trigraph $G_\iota$ of $\seg$:
\begin{itemize}
\item the subgraph induced by (the bags containing vertices of) $V(X_w)$ contains no black edges and all red arcs in it are symmetric, and
\item vertices of $X_v$ are in bags only with other vertices of $X_v$, and
\item  for each pair of canonically isomorphic vertices $z\in V(X_v)$ and $z'\in V(X_w)$, let $B(z)$ and $B(z')$ be the bags containing $z$ and $z'$ in $G_\iota$, respectively. Then $\alpha(\beta(B(z)))=\beta(B(z'))\cap V(X_w)$---in particular, the bag contents match when restricted to the subtrees rooted at $v$ and $w$. 
 \end{itemize}
 \item Each contraction in $\seg$ is a contraction between corresponding vertices of the subtrigraphs induced by (the bags containing vertices of) $V(X_v)$ and $V(X_w)$ that merges pairwise-isomorphic vertices w.r.t.\ $\alpha$.
\item In the last trigraph of $\seg$, each vertex of $X_v$ has been contracted with a vertex from $X_w$.
\end{itemize}
\end{definition}

Intuitively, an \emph{$X_v$-merge} $\seg$ is a subsequence of $\C$ such that at the beginning of the subsequence, $X_v$ and $X_w$ are in the ``same state'' if we ignore external vertices merged with the latter, and $\seg$ merely contracts the two subtrees into each other. The conditions on $X_w$ moreover guarantee that $\seg$ acts as a deletion operation on $X_v$ which will not result in new red edges in the resulting trigraph (even though new red edges may be created in the intermediate steps). In particular:

\begin{lemma}
\label{lem:otwwseg}
Let $\seg^\emph{start}$ and $\seg^\emph{end}$ be the first and last trigraphs in an $X_v$-merge \seg, respectively. Then $\seg^\emph{end}= \seg^\emph{start}-\{B(x)~|~x\in V(X_v)\}$ and, for each trigraph $\seg^\textup{mid}$ between $\seg^\emph{start}$ and $\seg^\emph{end}$ (included) in $\mathcal{C}$, $\otww(\seg^\textup{mid})\leq 2\cdot \otww(\seg^\emph{start})$.  \end{lemma}

\begin{proof}
Let us consider a trigraph $\seg^\textup{mid}$. Recall that with oriented twin-width, when merging two vertices together, the red out-degree of all other vertices can only decrease or stay the same, so we simply need to make sure that when merging a pair, the resulting vertex has red out-degree at most $2\cdot \otww(\seg^\emph{start})$. By definition of a segment, all contractions that we perform are disjoint. Combining these two facts, for each performed contraction, the highest possible red out-degree that it could create is if we start with this contraction---in particular, if a contraction creates a vertex with red out-degree $\zeta$, then performing the same contraction first in $Y$ would create a vertex with red out-degree at least $\zeta$. Thus, we only need to prove that performing any of the contractions in $Y$ first does not create a vertex with red out-degree exceeding the desired bound. 

To this end, observe that for such a pair $a_v, a_w \in V(\seg^\emph{start})$, where $a_v$ (resp. $a_w$) is in the subtrigraph induced by $X_v$ (resp. $X_w$) in $\seg^\emph{start}$, the neighborhood of $a_v$ can be partitioned into two categories: common neighbors with $a_w$, and vertices belonging to the subtrigraph induced by $X_v$. Indeed, this follows from the fact that $X_v$ and $X_w$ are twin-blocks.
For a neighbor $x$ of the first category, we can moreover say that if $a_v$ has a red arc to $x$, so will $a_w$; this follows from the fact that there is a pair of vertices from $G$ that lie in $a_v$ and are responsible for this red arc, and this pair has a corresponding ``mirrored'' pair of vertices from $G$ that lie in $a_w$ which will create a red arc as well.
 
For a neighbor $x$ of the second category, we know for the same reason that $a_w$ is a neighbor of the vertex $x'$ in the subtrigraph induced by $X_w$ which corresponds to $x$.
 Since all edges are red bi-directional arcs within the graph induced by $X_w$, $x'$ is a red neighbor of $a_w$.
Moreover, $a_w$ cannot have any black neighbor which is not shared with $a_v$. Because of these observations, when contracting $a_v$ and $a_w$, the red out-neighbors of the resulting vertex $a_\emph{new}$ are precisely the red (out-)neighbors of $a_w$ and the neighbors of $a_v$ in $V(X_v)$. However, this second category is in bijection with red (out-)neighbors of $a_w$, making the total red out-degree of 
 $a_\emph{new}$ upper-bounded by twice the red out-degree of $a_w$, i.e. by $2\cdot otww(\seg^\emph{start})$.

To observe that $\seg^\emph{end}=\seg^\emph{start}-\{B(x)~|~x\in V(X_v)\}$, it suffices to see that after each contraction $a_v$ with $a_w$, the resulting vertex $a_\emph{new}$ has exactly the same connections to $\seg^\emph{start}-\{B(x)~|~x\in V(X_v)\}$ as $a_w$ in $\seg^\emph{start}$, and additional red arcs to certain vertices of $\{B(x)~|~x\in V(X_v)\}$ which were not yet contracted in the segment. However, each such arc is mirrored by an arc from $a_\emph{new}$ to the corresponding vertex in the subgraph induced by $X_w$ in $\seg^\emph{start}$. In particular, after the last contraction in $Y$, all such additional red arcs will disappear and the 
  remaining connections for $a_\emph{new}$ are at that point identical to those of $a_w$ in $\seg^\emph{start}-\{B(x)~|~x\in V(X_v)\}$. 
\end{proof}

For each $X_v$-merge $Y$, we define its \emph{internal part} $Y^\text{int}$ as the subsegment obtained by removing the first and last trigraph of $Y$.
We now show that the above notion of $X_v$-merges has two useful properties: they are unique (in the sense of Lemma~\ref{lem:mergeunique}) and they do not interfere with each other (in the sense of Lemma~\ref{lem:mergeseparate}).

\begin{lemma}
\label{lem:mergeunique}
Let $G$ be a graph, $T$ a treedepth decomposition of $G$ and $\C$ a contraction sequence for $G$. Each vertex $v\in V(G)$ admits at most a single $X_v$-merge in $\C$.
\end{lemma}

\begin{proof}
Let $\seg$ be an $X_v$-merge segment in $\C$. Then the first contraction in $\seg$ (after $\seg^\emph{start}$) has to be the first contraction in $\C$ of a bag containing vertices of $X_v$ with a bag containing vertices outside of $X_v$. Thus, all potential $X_v$-merges start at the same trigraph. Moreover, the length of the segment is uniquely determined by the number of bags containing vertices of $X_v$ in $\seg^\emph{start}$. Thus, any two $X_v$-merges in $\C$ would need to start at the same trigraph and have the same length---meaning that they are not distinct.
\end{proof}

\begin{lemma}
\label{lem:mergeseparate}
Let $G$ be a graph, $T$ a treedepth decomposition of $G$, $\C$ a contraction sequence for $G$ and $v,v'\in V(G)$. If $v$ and $v'$ admit an $X_v$-merge $Y$ and an $X_{v'}$-merge $Y'$, respectively, then $Y_\text{int}$ and $Y'_\text{int}$ do not intersect in $\C$, unless the following holds: $X_v \sim X_{v'}$ and $\seg=\seg'$.
\end{lemma}

\begin{proof}
Let us consider any contraction happening during $\seg$. By the definition of an $X_v$-merge, it has to lie between a bag containing only vertices of $X_v$, and a bag containing no vertices of $X_v$, but at least some vertex of a twin-block of $X_v$. In the first case, the vertices of the second bag contain exclusively vertices from a twin-block $X_{v'}$ of $X_v$. This means that the only other merge segment that could possibly be happening during this contraction is an $X_{v'}$-merge $\seg'$. Indeed, any candidate needs to contain the vertices of the bag, but no vertices of the first bag (containing vertices of $X_v$), and only $X_{v'}$ satisfies both conditions. If this $\seg'$ indeed exists, then it is easy to verify that it has to start and end exactly as $\seg$, since then $\seg=\seg'$ consists entirely of contractions that identify bags containing vertices of $X_v$ and bags containing their equivalents in $X_{v'}$.

In the second case, there is no twin-block of $X_v$ containing all the vertices of the second bag. Since any vertex $v'$ such that $X_{v'}$ contains all the vertices of the bag would be an ancestor of $v$, we have $X_v\subset X_{v'}$, and thus the contraction is within $X_{v'}$, contradicting the fact that it would be part of a $X_{v'}$-merge.

Note that, if the internal parts of two segments were to intersect, this would imply the existence of at least three consecutive trigraphs both belonging to both segments. However, this would imply in particular that these two segments share a contraction, and we conclude since we analyzed precisely this case.
\end{proof}

In particular, we note that the only distinction between the trigraph at the end of $Y$ and the trigraph at the beginning of $Y$ is that all vertices of $T(v)$ have been removed. This will allow us to use $X_v$ merges as a viable reduction rule when parameterized by treedepth---however, care must be given to the fact that during the merge, the oriented twin-width may increase by a factor of at most $2$. Towards formally treating this, we define the following refinement of oriented twin-width that will serve as an invariant during the recursive application of our reduction rule.

\begin{definition}
Let $G$ be a graph, $T$ a treedepth decomposition of $G$, and $t$ a positive integer.
We say that a (potentially partial) contraction sequence $\C$ for $(G, T)$ has \emph{$(t, 2t)$-oriented twin-width} if:
\begin{itemize}
\item the width of $\C$ is at most $2t$, and
\item each trigraph in $\C$ of width larger than $t$ lies in the internal part of an $X_v$-merge for some $v\in V(G)$.
 \end{itemize}
\end{definition}

\begin{redrule}\label{red:rr}
Let $(G, T)$ be a graph with its nice treedepth decomposition of depth $p$.
Let $u\in T$ be not yet processed such that $\forall v\in V(T_u)\setminus \{u\}$, $v$ is processed. Let $d$ be the depth of $u$. For a connected component $H$ in $\con_u$ such that the equivalence class of $H$ for $\sim_u$ is strictly larger than $h(d,p)=2^{2^{4p}\cdot b(d+1,p)}$, delete $H$ from both the graph and the treedepth decomposition to obtain a pair $(G', T')$.
\end{redrule}

Observe that by the choice of $H$, if $G$ was connected then the graph $G'$ obtained by \cref{red:rr} remains connected. 
The cornerstone of our proof will be the following lemma, which states that \cref{red:rr} is a safe reduction rule to use and that we prove in \cref{ss:mainlemma}. 

\begin{lemma}\label{lem:rr}
Let the pair $(G, T)$ (resp.\ $ (G',T')$) contain the graph and the associated decomposition before (resp.\ after) applying Reduction Rule~\ref{red:rr}. If $(G', T')$ admits a $(t, 2t)$-oriented twin-width contraction sequence $\C'$, then $(G, T)$ does as well and we can construct the corresponding sequence $\C$ from $\C'$ in time $2\uparrow \uparrow \bigoh(p) \cdot |V(G)|^{\bigoh (1)}$.
\end{lemma}

For now, we proceed towards establishing our results conditioned on the correctness of Lemma~\ref{lem:rr}. The following statement summarizes the outcome after the exhaustive application of Reduction Rule~\ref{red:rr}.

\begin{lemma}
\label{lem:exhaust}
Let $(G, T)$ be a graph and its nice treedepth decomposition with depth $p$.
Applying Reduction Rule~\ref{red:rr} exhaustively can be done in polynomial time, and in the resulting pair $(G^*, T^*)$ the graph $G^*$ has size upper bounded by $(g(1,p)+1)^p=2\uparrow\uparrow \bigoh(p)$.
Moreover, if $(G^*,T^*)$ has a $(t, 2t)$-oriented twin-width contraction sequence $\C^*$, $(G, T)$ does as well and we can construct the corresponding sequence $\C$ from $\C^*$ in time at most $2\uparrow \uparrow \bigoh(p) \cdot |V(G)|^{\bigoh (1)}$.
\end{lemma}

\begin{proof}
Let $(G,T)$ be a graph and its nice treedepth decomposition of depth $p$. 
Let $(G^*, T^*)$ be the pair obtained after applying \cref{red:rr} exhaustively. First of all, each application of the reduction rule deletes some vertices from the graph, thus we know that it is applied at most $n$ times. We will prove by contradiction that all vertices of $(G^*, T^*)$ are processed. Let $v$ be a vertex in $T^*$ that is not processed such that all other vertices in $T^*_v$ are processed. Note that any vertex achieving maximal depth among the not processed vertices satisfies this condition. Let us note $d$ the depth of $v$. 

By \cref{lem:size_processed}, all subgraphs of $G^*$ corresponding in $T^*$ to the subtrees rooted at children of $v$ have size at most $b(d+1,p)=(g(d+1, p)+1)^p$ since their roots are processed and at depth $d+1$. Moreover, these same subgraphs have connections to at most $p-d$ vertices outside of themselves: by definition of $T^*$, they can only be adjacent to vertices on the path from the root of $T^*$ to $v$ (included). 
The number of equivalence classes for $\sim_v$ is then at most $y(d,p)\coloneq 2^{(p-d)\cdot b(d+1, p)} \cdot b(d+1, p) \cdot 2^{b(d+1, p)^2}$: the first term comes from the possible connections to the rest of the graph, the rest upper bounds the number of non-isomorphic graphs of size up to $b(d+1, p)$, see~\cref{sec:prelims}.
If \cref{red:rr} cannot be applied further, each of these equivalence class has at most $h(d,p)=2^{2^{4p}\cdot b(d+1,p)}$ elements. Multiplying these two values together to compute the total number $N_v$ of children of $v$ in $T^*$, we obtain $N_v \leq g(d,p)$.

Indeed, let $x\coloneq g(d+1, p)$, and let us rewrite the computation as follows: $y(d,p)= 2^{(p-d)\cdot(x+1)^p}\cdot (x+1)^p \cdot 2^{(x+1)^{2p}}$ and $h(d, p)=2^{2^{4p}\cdot (x+1)^p}$. We observe that, trivially upper-bounding $(x+1)^p$ by $2^{p\cdot (x+1)}$, we obtain $N_v\leq 2^{(x+1)^p \cdot ((p-d+1) + (x+1)^p +2^{4p})}$. We then upper bound separately $p-d+1 \leq p \leq (x+1)^{4p}$, $(x+1)^p\leq (x+1)^{4p}$ and $2^{4p}\leq (x+1)^{4p}$, which gives the following: $N_v \leq  2^{3(x+1)^{5p}}\leq 2^{(x+1)^{6p}}= g(d, p)$, with the last transformation relying on $x\geq 2$ which follows from \cref{def:branchingfactor}. 

The above yields the desired contradiction to $v$ not being processed, since it exactly satisfies all conditions listed in \cref{def:branchingfactor}. We have thus proved that all vertices of $(G^*, T^*)$ are processed, and in particular that $G^*$ contains at most $2 \uparrow \uparrow \bigoh(p)$ vertices. 

Finally, we use \cref{lem:rr} on each application of \cref{red:rr} in a \emph{first-in last-out} order to construct from $\C^*$ a construction sequence $\C$ with $(t, 2t)$-oriented width for $(G,T)$. The running time bound follows from \cref{lem:rr} as well.
\end{proof}

Using Lemma~\ref{lem:exhaust} to guarantee the desired properties during the exhaustive application of Reduction Rule~\ref{red:rr}, we obtain the claimed tractability result for oriented twin-width:

\begin{theorem} 
\label{thm:param-by-TD}
$2$-approximating oriented twin-width is \FPT parameterized by treedepth. Specifically, given a graph $G$ on $n$ vertices and treedepth $p$, there exist a computable function $\lambda$ and an algorithm running in time $\lambda(p)\cdot n^{\bigoh(1)}$ which outputs a contraction sequence for $G$ of oriented width at most $2\cdot\otww(G)$.
\end{theorem}

\begin{proof}
Let $G$ be a graph on $n$ vertices, and with treedepth $p$.
We compute a nice treedepth decomposition $T$ of depth $p$ for $G$ in time $2^{\bigoh(p^2)}\cdot n$, see~\cite{ReidlRVS14}.
We then apply exhaustively Reduction Rule~\ref{red:rr} on the pair $(G, T)$ and obtain in polynomial time a pair $(G^*, T^*)$ such that $G^*$ has at most $n^*=2\uparrow \uparrow \bigoh(p)$ vertices, see Lemma~\ref{lem:exhaust}.

Now, we use Fact~\ref{fact:computecontract} to compute a contraction sequence $\C^*$ of oriented width $t^*$ for $G^*$ in time at most $2^{\bigoh(n^* \cdot \log(n^*)}=2\uparrow \uparrow \bigoh(p)$, and then use again Lemma~\ref{lem:exhaust} to create in polynomial time a contraction sequence $\C$ for $G$. Note that the lemma ensures that $\C$ is a $(t^*, 2t^*)$-oriented twin-width contraction. In particular, $\C$ has oriented width at most $2t^*$, and since $G^*$ is an induced subgraph of $G$, we have  that $t^*\leq t$ and $\C$ has oriented width at most $2t$.

Finally, it is easy to observe that the total running time is at most $\lambda(p)\cdot n^{\bigoh(1)}$, with $\lambda(p)=2\uparrow \uparrow \bigoh(p) \cdot$.
\end{proof}

At last, we are ready to prove our first main result:

\begin{theorem} 
\label{thm:mainone} 
There is an algorithm which takes as input an $n$-vertex graph $G$ of twin-width $t$ and treedepth $p$, runs in time at most $f(p)\cdot n^{\bigoh(1)}$
  and outputs a contraction sequence of width at most $q(t)$, for some computable functions $q$, $f$.
\end{theorem}

\begin{proof}
We first apply \cref{thm:param-by-TD} to obtain a contraction sequence $\C_o$ of oriented width $2 \cdot \otww(G)$. This takes $\lambda(p)\cdot n^{\bigoh(1)}$ time, and we can then use \cref{lem:translation} on $G$ and $\C_o$ together to obtain a contraction $\C$ for $G$ of width at most $q(t)=2^{2^{2^{\bigoh(t)}}}$. Since the transformation in \cref{lem:translation} runs in \FPT time parameterized by $p$, there is a computable function $\zeta$ such that it runs in $\zeta(t) \cdot n^{\bigoh(1)}$ time. Moreover, $t$ is upper bounded by a function of $p$~\cite{BonnetKTW22} and so there exists a function $\mu$ such that the transformation takes time at most $\mu(p)\cdot n^{\bigoh(1)}$.
The total running time is then simply $f(p)\cdot n^{\bigoh(1)}$, for $f(p)=\lambda(p)+\mu(p)$.
\end{proof}

\subsection{Proof of \cref{lem:rr}} \label{ss:mainlemma}

The main idea to prove \cref{lem:rr} is that we will integrate the deleted subgraph $H$ into $\C'$ in a very local manner. The high-level ideas underlying this integration are inspired by those previously used for approximating twin-width via vertex integrity~\cite{BalabanGR24IPEC}, but with significant distinctions caused by the difference in both the employed parameter (treedepth) and the computed measure (oriented twin-width).

Intuitively, we will first follow the sequence of contractions $\C'$ without interacting with $H$, until at some point we put $\C'$ on hold, identify a special twin-block $H'\sim H$, and \emph{insert} $H$ using $H'$. Namely, we will do contractions within $H$ to fit the current state of $H'$, and then use a $H$-merge to identify vertices of $H$ with vertices of $H'$. After this insertion, we can continue with the contractions of $\C'$ for the rest of the contraction. The main difficulties of this subsection are to identify the point where we stop following $\C'$ to insert $H$, and prove the existence of a twin-block $H'$ such that the obtained sequence satisfies our requirements.

Let the pair $(G, T)$ (resp.\ $ (G',T')$) contain the graph and the associated decomposition before (resp.\ after) applying Reduction Rule~\ref{red:rr}, and $H$ the subgraph of $G$ appearing in $T$ which was deleted to obtain $(G',T')$.
Let $\C'=(G'_0, G'_1, \dots)$ be a contraction sequence for $(G',T')$.

\begin{definition} \label{def:extension}
For a given trigraph $G'_i$ in $\C'$, the \emph{extension} $\ext(G'_i, H)$ is the trigraph obtained by applying the partitioning of vertices of $V(G')\cup V(H)$ following on $V(G')$ the partitioning corresponding to $G'_i$, and leaving all vertices in $V(H)$ in bags as singletons. This intuitively corresponds to the trigraph obtained from $G$ following the same contractions leading from $G'$ to $G'_i$.
We extend this notion to partial contraction sequences: $\ext(\C', H)=(\ext(G'_0, H), \ext(G'_1, H), \dots)$.
\end{definition}

Observe that $\ext(\C', H)$ and $\C'$ have the same length, and at the end of $\ext(\C', H)$, all vertices of $G'$ are merged into one bag, but the vertices of $H$ are still each in a separate bag.

\begin{definition}\label{def:safe}
Let the pair $(G, T)$ (resp.\ $ (G',T')$) contain the graph and the associated decomposition before (resp.\ after) applying Reduction Rule~\ref{red:rr}, $H$ be the subgraph of $G$ appearing in $T$ which was deleted to obtain $(G',T')$, and $\ext(\C', H)=(G=G_0, G_1, G_2, \dots)$ the extension of $\C'$ to $G$.
\begin{itemize}
\item We say that a trigraph $G_i$ in $\ext(\C', H)$ is $H$\emph{-indifferent} if there is no red arc from the bags containing vertices of $G'$ to the vertices of $H$.
\item We say that a trigraph $G_i$ in $\ext(\C', H)$ is $H$\emph{-safe} if: there exist $H',H''$ twin-blocks to $H$ which are \emph{merged} by $\C'$ in $G_i$---in particular, for each $u\in V(H')$, there is a vertex $v \in V(G_i)$ such that $\{u, \alpha(u) \}\subseteq \beta(v)$.
 \end{itemize}
\end{definition}

\begin{observation}
$\ext(\C', H)$ contains $H$-indifferent trigraphs as well as trigraphs which are not $H$-indifferent. Moreover, the property is monotone: along $\ext(\C', H)$ if a trigraph is $H$-indifferent then so is every trigraph preceding it in $\ext(\C', H)$. 
\end{observation}

\begin{proof}
$G_0=G$ is obviously $H$-indifferent, since it has no red arcs. On the other hand, the last trigraph of $\ext(\C', H)$ is not $H$-indifferent. Indeed, $G$ is connected, thus there is a vertex in $V(G')$ connected to a vertex $u$ of $H$. However, there exist twin-blocks of $H$ in $G'$, implying the existence of vertices in $G'$ not connected to $u$ in $G$ by definition of a treedepth decomposition. Thus, at the end of $\ext(\C', H)$, the single bag containing all vertices of $G'$ has a red arc to the bag containing (only) $u$. 

For the claimed monotonicity property, we prove the contrapositive. There, it suffices to observe that a red arc can only disappear completely if both of its endpoints are merged together. Thus, as soon as there is a red arc from the non-$H$ part of the trigraph to the $H$ part of it, it is preserved by the contractions of $\ext(\C', H)$, which happen exclusively within the non-$H$ part of the trigraph. Note that several such red arcs to $H$ can be merged together by these operations, but this is not important here as we only need the existence of a single one to prevent $H$-indifference. 
\end{proof}

The next lemma guarantees an $H$-safe trigraph sufficiently early in $\ext(\C', H)$.

\begin{lemma}
\label{lem:merge}
Let the pair $(G, T)$ (resp. $ (G',T')$) contain the graph and the associated decomposition that occurs before (resp. after) applying Reduction Rule~\ref{red:rr}, $H$ be the subgraph of $G$ appearing in $T$ which was deleted to obtain $(G',T')$, and $\ext(\C', H)=(G=G_0, G_1, G_2, \dots)$ the extension of $\C'$ to~$G$.
The first trigraph in $\ext(\C', H)$ which is not $H$-indifferent is $H$-safe.
\end{lemma}

\begin{proof} 
Let $G_{i-1}$ be the last $H$-indifferent trigraph in $\ext(\C', H)$. Suppose that $G_i$ is not $H$-safe. By definition of $G_i$, there exists $u\in V(H)$ and $v\in V(G_i)\setminus V(H)$ such that $vu$ is a red arc. Let  $\delta$ be the depth of the seed $s_H$ of $H$ in $T$. Let $d=2^{p+2}+1$. Let $I=[h(\delta,p)]$ 
and let $H_1,\dots, H_{h(\delta,p)} \in [H]_\sim$ be distinct twin-blocks of $H$ present in $G'$. We can ensure the existence of these because the equivalence class of $H$ was sufficiently large in $G$ to apply \cref{red:rr}, and in particular still remains large in $G'$. We now prove the following claim by induction:

\begin{claim}\label{cl:merge}
Let $p_H \coloneq p-\delta$.
For each $a\in [0, 2^{p_H}-1]$, there is a set $I_a\subseteq I$ of size at least $h(\delta,p)/d^{a\cdot b(\delta)+1}$ such that for each $j, k\in I_a$ and each vertex $w\in V(H)$ at distance at most $a$ from $u$ in $H$, there is a vertex $x\in V(G_i)$ such that $w_j$ and $w_k$ both belong to the bag $x$.
\end{claim}

\begin{proof}[Proof of \cref{cl:merge}]
Let us start by proving the claim for $a=0$. Since $vu$ is a red arc, $v$ is a bag containing both a neighbor $y$ and a non-neighbor $z$ of $u$ in $G$, and by definition of $G_i$ they are both outside of $H$. By definition of the treedepth decomposition, $y$ has to be an ancestor of $s_H$ in $T$ while $z$ can be in any part of $G'$. Observe that, by \cref{def:equivalence}, for each $j\in I$, $y$ is a neighbor of $u_j$ in $G'$, and $z$ a non-neighbor of $u_j$ unless $z\in V(H_j)$. Thus, for all but at most one $j$ in $I$, the descendant $u'_j$ of $u_j$ in $G_i$ is a red neighbor of $v$. However, $v$ has a number of red neighbors upper bounded by $2^{p+2}$. Indeed: the red out-degree of $v$ is at most twice the oriented twin-width, which is at most the twin-width, which is itself upper-bounded by a function of the treewidth, and lastly the treewidth is at most the treedepth.

Hence, the vertices of $U \coloneq \{u_j, j\in I\}$ are present in the bags of at most $d$ vertices in $G_i$ (note that some might be in the bag of $v$), which means that there is a vertex $w\in V(G_i)$ with at least $h(\delta,p)/d$ vertices of $U$ in its bag. Now it suffices to set $I_0 \coloneq \{j\in I, u_j \in \beta(w)\}$. 
This concludes the proof of the base case of the induction.

For the induction step, suppose that \cref{cl:merge} holds for some $a\in [0, 2^{p_H}-2]$, \ie there is a set $I_a\subseteq I$ with the desired properties. Note that we can assume $p_H \geq 2$, otherwise the proof is already complete. Let $D_a, D_{a+1}\subseteq V(H)$ be the sets of vertices at distance exactly $a$ or $a+1$ from $u$ in $H$, respectively. 
Let $w \in D_{a+1}$ and $x \in D_a$ be two adjacent vertices in $H$, and let $W = \{w_j \sep j \in I_a\}$.
For each $j \in I_a$, let $x'_j$ be the descendant of $x_j$ in $G_i$ 
and let $w_j'$ be the descendant of $w_j$ in $G_i$. 
By the induction hypothesis, $x_j' = x_k'$ for each $j,k \in I_a$; let $x' := x'_j$ for some/each $j \in I_a$.
Observe that, for each $k \in I_a$, $x'w_k'$ is a red edge of $G_i$ or $x' = w_k'$. Using the same argument as in the base case, $x'$ has red degree at most $d-1$ in $G_i$, which means that the vertices of $W$ are present in the bags of at most $d$ vertices in $G_i$.

Since $|D_{a+1}| \le b(\delta)$, $I_a$ can be partitioned into at most $d^{b(\delta)}$ parts such that if $j,k \in I_a$ are in the same part, then for each vertex $w \in D_{a+1}$, $w_j$ and $w_k$ are in the bag of the same vertex in $G_i$. Hence, one of these parts has size at least $|I_a| / d^{b(\delta))}$, and we choose it to be $I_{a+1}$. A simple computation shows that $I_{a+1}$ satisfies Claim~\ref{cl:merge}.
\end{proof}

Observe that \cref{cl:merge} implies that $H_j$ and $H_k$, for any $j,k\in I_{2^{p_H}-1}$, are merged in $G_i$ because the diameter of $H$ is at most $2^{p_H}-1$. Indeed, the depth of the tree associated to $H$ in $T$ is at most $p_H$, and the diameter of a connected graph of treedepth $p_H$ is at most $2^{p_H}-1$.
Additionally, we need to verify that $|I_{2^{p_H}-1}| \geq h(\delta,p)/d^{(2^{p_H}-1)\cdot b(\delta)+1} \geq 2$.
 Recall that $h(\delta,p)=2^{2^{4p}\cdot b(\delta)+1}$ and $d=2^{p+2}+1\leq 2^{3p}$. We use the fact that $(2^{p_H}-1)\cdot b(\delta)+1 \leq 2^{p+1} \cdot b(\delta)$ to get $d^{(2^{p_H}-1)\cdot b(\delta)+1}\leq 2^{3p\cdot 2^{p+1} \cdot b(\delta)} \leq 2^{2^{4p} \cdot b(\delta)}$, which concludes the proof.
\end{proof}

In fact, we need a slightly stronger result: for our purposes, it will be necessary to have a graph which is both $H$-indifferent and $H$-safe. 

\begin{lemma}
\label{lem:Safe}
Let the pair $(G, T)$ (resp.\ $ (G',T')$) contain the graph and the associated decomposition before (resp.\ after) applying Reduction Rule~\ref{red:rr}, $H$ the subgraph of $G$ appearing in $T$ which was deleted to obtain $(G',T')$, and $\ext(\C', H)=(G=G_0, G_1, G_2, \dots)$ the extension of $\C'$ to $G$.
The last $H$-indifferent trigraph $G^*$ in $\ext(\C', H)$ is $H$-safe.
\end{lemma}

\begin{proof}
Let us consider the last $H$-indifferent trigraph $G_{i-1}$. The next trigraph $G_i$ in $\ext(\C', H)$ is $H$-safe as per \cref{lem:merge}, so there exist two graphs $H', H'' \in [H]_\sim$ merged in $G_i$.
Suppose by contradiction that $H'$ and $H''$ are not yet merged together in $G_{i-1}$. Let $u, v$ the two vertices of $G_{i-1}$ which are contracted together to obtain $G_i$. Since their merge needs to finalize the merge of $H'$ and $H''$, we know that each of their bags contain a vertex of $H'\cup H''$. Thus, since $H'\cup H''$ is not in the neighborhood of $H$, and that there is no red arcs to $H$, it means that the bags of $u$ and $v$ can contain only vertices not adjacent to $H$. However, this implies that the contraction of $u$ with $v$ cannot create a red arc to any vertex of $H$, and $G_i$ would thus be a $H$-indifferent trigraph, which contradicts the definition of $G_i$.  
\end{proof}

The following lemma has a crucial role in controlling the approximation factor of the final algorithm.

\begin{lemma}
\label{lem:notinmerge}
Let the pair $(G, T)$ (resp.\ $ (G',T')$) contain the graph and the associated decomposition before (resp.\ after) applying Reduction Rule~\ref{red:rr}, $H$ the subgraph of $G$ appearing in $T$ which was deleted to obtain $(G',T')$, and $\ext(\C', H)=(G=G_0, G_1, G_2, \dots)$ the extension of $\C'$ to $G$.
The last $H$-indifferent trigraph $G^*$ in $\ext(\C', H)$ does correspond in $\C'$ to a trigraph $G'_{i-1}$ which is not in the internal part of any $X_u$-merge segment.
\end{lemma}

\begin{proof}
Let us prove this by contradiction. Suppose there exist, for some $u\in V(G')$, an $X_u$-merge $\seg$ of twin-blocks $A\sim A'$ such that $G'_{i-1}$ is in the internal part of $\seg$. This implies that $G'_i$ is also part of $\seg$ (possibly $\seg^\emph{end}$). Thus, the two vertices $u,v\in V(G'_{i-1})$ merged together to obtain $G_i$ both contain in their bags vertices of $A \cup A'$. We observe that no vertex in $H$ is the neighbor of a vertex of $A$: by definition of $H$, it can only have connections to the vertices on the path from the root to $H$ in $T$, and if a vertex of $A$ (resp. $A'$) was on this path, it would imply that $H\subseteq A$ (resp. $H\subseteq A'$), which contradicts the fact that $A$ and $A'$ appear in $G'$. We now can conclude similarly to the proof of \cref{lem:Safe}.
\end{proof}

With this final intermediate step done, we proceed to the proof of~\cref{lem:rr}.

\begin{proof}[Proof of \cref{lem:rr}]
We first construct $\ext(\C', H)=(G=G_0, G_1, \dots)$, the extension of $\C'$ to $G$, which can be done in polynomial time. We moreover identify $i$ the index of the first trigraph in $\ext(\C', H)$ which is not $H$-indifferent.
By \cref{def:safe} and \cref{lem:Safe}, there are $H', H'' \in [H]_\sim$ that are merged in $G_{i-1}$. Towards computing these, we first test, for each pair of siblings $a,b$ of the seed $z$ of $H$ in $T$, whether $X_a$ and $X_b$ are twin-blocks (i.e., belong to the same equivalence class of $\sim$). If they are, we have a canonical isomorphism $\alpha:H'\rightarrow H''$ witnessing the equivalence, and can then test whether each vertex $v$ of $H’$ lies in the same bag as $\alpha(v)$ in $G_{i-1}$. Once again, by~\cref{lem:Safe} we are guaranteed that this test will eventually succeed in time at most $\bigoh(|H|^{|H|}\cdot n^2)$.

 Let $\C'_\text{indifferent}:= (\C')_{[0, i-1]}$. Let $\C'_{H'}$ be the restriction of $\C'_\text{indifferent}$ to $H'$, i.e., $\C'_{H'}=\C'_\text{indifferent}[H']$. Using the canonical isomorphism $\alpha$ from $H'$ to $H$, we define $\C_H$ as the result of applying $\alpha$ on each vertex (or bag) of each trigraph in $\C'_{H'}$. It is easy to observe that $\C_H$ is a valid partial contraction of $H$.\\
Let us now define $4$ partial contraction sequences:
\begin{itemize}
\item$\C_A\coloneq \ext(\C'_\text{indifferent}, H)$,
\item$\C_B\coloneq \ext(\C_H, G'_{i-1})$,
\item$\C_C$ is an arbitrary $H$-merge of $H$ into $H'$ starting from the last trigraph of $\C_B$,
\item $\C_D\coloneq (\C')_{\geq i-1}$ the end of the contraction sequence $\C'$.
\end{itemize}

Let us first observe that $\C_A$ starts with the trigraph $G$ and ends with the trigraph $\ext(G'_{i-1}, H)$. $\C_B$ starts with $\ext(H, G'_{i-1})$ which is exactly $\ext(G'_{i-1}, H)$. $\C_C$ starts precisely with the last trigraph of $\C_B$, and ends by definition with $G'_{i-1}$. $\C_D$ starts with $G'_{i-1}$ and ends with $K_1$, the trigraph on one vertex.\\
We now combine them, appending them in order after removing the last trigraph in $\C_A$, $\C_B$ and $\C_C$ (to avoid duplicates). The resulting sequence of trigraphs $\C$ is a contraction sequence of $G$. Indeed, by construction of $\C_A, \C_B, \C_C$ and $\C_D$, one can go from one trigraph of $\C$ to the next by a valid contraction of two vertices. Moreover, the oriented width of $\C$ is bounded, and we will show it separately on each segment. 
On the first part $\C_A$, it suffices to see that it has $(t, 2t)$-oriented width if and only if $\C'_\text{indifferent}$ has $(t, 2t)$-oriented width. Indeed, the extension with $H$ does not add any red arcs to any of the trigraphs by definition of $H$-indifference. $\C'_\text{indifferent}$ is a subsequence of $\C'$, and neither its first trigraph, \ie $G'$, nor its last, \ie $G'_{i-1}$, are in the internal part of merge segments, see \cref{lem:notinmerge}, which suffices to conclude that it satisfies the requirements. Note that the same directly applies to $\C_D$ (without the extension part of the argument, here unnecessary).
Let us take a look at $\C_B$. Because $G'_{i-1}$ is not in the internal part of a merge, we know that it has maximum red out-degree $t$, and the first trigraph of $\C_B$ as well. During $\C_B$, all contractions happen between vertices of $H$, thus the red out-degree of all vertices not in $H$ can only decrease. By definition of $\C_H$, for each trigraph $G_\mu$ of $\C_B$, it is possible to find a trigraph $G'_j$ in $\C'$ such that $H'$ was at the exact same level of processing as $H$ in $G_\mu$. We claim that the red out-degree of vertices of $H$ in $G_\mu$ cannot be larger as the red out-degree of the corresponding vertices of $H'$ in $G_j$. \\

First observe that the neighborhood of $H$ in $G$ has never been merged with $H'$, since it would otherwise have created red arcs to $H$ and would contradict the choice of $G_{i-1}$. Now, let us partition $V(G')$ into the vertices merged with the vertices of $H'$ in $G_{i-1}$, and the rest $R$. Observe that all red arcs between vertices of $H$ and $V(G')$ are to $R$, since $N(H)\subseteq R\cup H$. The red arcs present from a vertex $v\in H$ in $G_\mu$ to other vertices of $H$ have one-to-one equivalents from $v_{H'}$ to other vertices of $H'$ in $G_j$. The other red arcs starting in $v$ are to a subset of the remaining vertices from $R$, and this is a subset of the red neighborhood of the bag of $H'$ in $G_\mu$ containing $v_{H'}$, since red arcs can only disappear when both endpoints are contracted, which is forbidden by the definition of our partitioning. Thus, we can check that the trigraphs of $\C_H$ have maximum red out-degree $2t$, and the trigraphs having maximum red out-degree over $t$ correspond to trigraphs of $\C'$ having maximum red out-degree over $t$ because of merge segments of parts in $H'$, meaning that they are in the internal part of merge segments of parts of $H$. We note moreover that, since $G_{i-1}$ is not in the internal part of any merge segment, it holds in particular for merge segments of parts of $H'$, and yields that the end of $\C_B$ is also not internal to any segment. We have thereby proved the requirements of $(t,2t)$-oriented width. 

Lastly, we observe that $\C_C$ starts with a trigraph of maximum red out-degree at most $t$ because we proved that $\C_B$ has $(t,2t)$-oriented width. Moreover, $\C_C$ corresponds to a $H$-merge segment by \cref{def:mergeseg}. Thus, by \cref{lem:otwwseg}, the maximum red out-degree of each internal trigraph is at most $2t$, and of its last trigraph at most $t$ since it is an induced subtrigraph of its first trigraph.

All the pieces of $\C$ can be computed in time upper-bounded by $|H|^{|H|}\cdot n^{\bigoh(1)}$,
 thus the computation of $\C$ itself can be upper-bounded by the same function.
\end{proof}

\section{An Exact Algorithm Based on Vertex Integrity}
\label{sec:vi}

Let $S$ be a vertex set satisfying that for each connected component $H$ of $G - S$, $|V(H) \cup S| \leq \vi(G)$, and recall that we may compute such a set $S$ using known results (see Section~\ref{sec:prelims}). Let $\mathcal{D}_S$ denote the set of connected components of $G-S$; we drop the $S$ when the set is clear from context. For the rest of the section, we assume to have computed and fixed a specific choice of $S$ and that $p=\vi(G)$. Throughout the section, we assume that $p\geq 2$ since instances where $p\leq 1$ are trivial.

\subsection{Setting Up the Framework}
In this subsection, we adapt the preparatory steps for using the framework~\cite{DFGS25} to our setting. Our aim here is to define and compute a ``reduced graph'' whose size is bounded by a function of $\vi(G)$ (Lemma~\ref{lem:reducedcomputation}). The main novel contributions then lie in Subsection~\ref{sub:lifting}, which shows that a contraction sequence of the reduced graph can be lifted to a contraction sequence of the original graph.

We first define a basic notion of equivalence between components and restate a known result about computing these.

\begin{definition} We say two graphs $H_0, H_1 \in \mathcal{D}$ are \emph{twins}, denoted $H_0 \sim H_1$, if there exists a canonical isomorphism $\alpha$ from $H_0$ to $H_1$ such that for each vertex $u \in V(H_0)$ and each $v \in S$, $uv \in E(G)$ if and only if $\alpha(u)v \in E(G)$. 
 \end{definition}

\begin{fact}[\cite{DFGS25}]
\label{fact:size-equiv-class}
          Each graph $H\in \mathcal{D}$ has at most $p$ vertices, $\sim$ is an equivalence relation and the number of equivalence classes in $[\sim]$ is upper-bounded by $p \cdot 2^{2p^2}$. Moreover, a partition of connected components into $[\sim]$ can be computed in time at most $\mathcal{O}(p \cdot 2^{2p^2}\cdot n)$.
\end{fact}

A crucial feature of the Ramsey Pruning technique is that it requires every connected component in $\mathcal{D}_S$ to occur sufficiently many times. These components will then later be placed into groups, as we define below.

\begin{definition}
\label{defn:large-eq-class}
Let $k$ be a positive integer, an equivalence class $[H]_\sim$ of $\sim$ is said to be {\em $k$-large} if $|[H]_\sim|\geq k$. Further a vertex set $L \subseteq V(G)$ is called a {\em $k$-large group} if the induced subgraph $G[L]$ is a disjoint union of exactly one graph from each $k$-large equivalence class of~$\sim$.
\end{definition}

With this, we can proceed to a formalization of our reduced graphs. Intuitively, this is the graph obtained by first recursively adding all equivalence classes that are ``too small'' for our purposes into $S$, until we form a deletion set $S\cup S'$. Note that adding components to the deletion set in such a way does not impact the remaining components, nor the twin relation between them. After that, we place a parameter-bounded number of representatives of ``large'' equivalence classes into groups and delete all the remaining connected components in $\mathcal{D}_{S\cup S'}$.
Towards formalizing this, we fix two computable functions: $g(p)=2\uparrow\uparrow(p\cdot 2^{2p^2}\cdot 6+6)^p$ upper-bounds the size of the deletion set after we expand it to contain all small equivalence classes, and $f(p,x)=2^{2^{2^{2^{2^{2^{xp+1}}}}}}$ specifies a safe bound for how large the groups need to be.
   
\begin{definition}
\label{def:red-graph}
An induced subgraph $G'$ of $G$ is said to be a {\em reduced graph} of $G$ if there exists a positive integer $x\leq g(p)$ and a partition of $V(G') = S \uplus S' \uplus Y$ satisfying:
\begin{enumerate}
    \item $|S\cup S'|= x$ and \( S' \) is the set of all vertices in graphs in \( \mathcal{D} \) that do not belong to an $f(p,x)$-large equivalence class of $\sim$.
    \item If there are no $f(p,x)$-large equivalence classes of $\sim$, $Y=\emptyset$ and $V(G')=S\uplus S' = V(G)$. Otherwise, it holds that $Y=L_1 \uplus \cdots \uplus L_{f(p,x)}$ where \( L_i \) 
         is an $f(p,x)$-large group for each $i\in\{1,\cdots,f(p,x)\}$, and each pair of $L_i$ and $L_j$, $i \neq j$, is vertex-disjoint.
\end{enumerate}
\end{definition}

We now prove that a reduced graph can be computed efficiently by essentially following the intuitive description in the previous paragraph. One technical challenge that is treated in the proof is that when we increase the size of the set $S\cup S'$, we also increase the lower bound for our large equivalence classes.

\begin{lemma}
\label{lem:reducedcomputation}
    There exists a reduced graph $G'$ of $G$, and given $S$ and $\sim$ such a graph can be computed in polynomial time. Further, the number of vertices in $G'$ is upper-bounded by $2\uparrow\uparrow((p\cdot 2^{2p^2}\cdot 6+6)^p+10)$.
\end{lemma}

\begin{proof}
We present an algorithm to compute a reduced graph $G'$ below:

\begin{enumerate}
    \item Initialize $X=S$ and $\mathcal{D}':=\mathcal{D}$.
    \item As long as there exists an equivalence class $[H]_\sim$, $H\in \mathcal{D}'$, that is not $f(p,|X|)$-large:
    \begin{enumerate}
        \item Set $\mathcal{D}'=\mathcal{D}'\setminus [H]_\sim$ and $X=X\cup \{v:v\in H', H'\in [H]_\sim\}$.
    \end{enumerate}
    \item If $\mathcal{D}'\neq \emptyset$, then set $O:=X\uplus L_1\uplus \cdots \uplus L_{f(p,|X|)}$ where each $L_{1\leq i\leq f(p,|X|)}$ is a $f(p,|X|)$-large group. Else set $O:=X$.
    \item Output $G'=G[O]$.
\end{enumerate}

We first bound the size of $V(G')=O$. In each iteration of step~$2$, the vertices of all graphs in exactly one equivalence class $[H]_\sim$, $H\in \mathcal{D}'$, are added to the set $X$ and the graphs in $[H]_\sim$ are removed from $\mathcal{D}'$. Let $X_i$, $i\in \{1,\cdots,2^{2p^2}\}$ be the set $X$ at the end of the $i^{th}$ iteration of step $2$ and let there be $t$ iterations of step $2$. Note that $t=0$ if there is no successful iteration of step 2. Since there are at most $p\cdot 2^{2p^2}$ equivalence classes in $\sim$ by Fact~\ref{fact:size-equiv-class}, the algorithm repeats step~2 at most $p\cdot 2^{2p^2}$ times, thus $t\leq p\cdot 2^{2p^2}$.

Let $X_0:=S$, observe that $|X_0|=|S|\leq p$. If $t\geq 1$, by construction and by the fact that each graph in $\mathcal{D}_S$ has at most $p$ vertices, for each $i\in \{1,\cdots,t\}$, $|X_i|\leq |X_{i-1}|+p\cdot f(p,|X_{i-1}|)$. Since $t\leq p\cdot 2^{2p^2}$, $|X_t|\leq g(p)$.
   
Further, each $f(p,|X_t|)$-large group has at most $p\cdot2^{2p^2}$ vertices, and thus $f(p,|X_t|)$ many $f(p,|X_t|)$-large groups have at most $p\cdot 2^{2p^2}\cdot f(p,|X_t|)$ vertices. Therefore we have $|O|\leq |X_t|+p\cdot 2^{2p^2}\cdot f(p,|X_t|)\leq g(p)+p\cdot 2^{2p^2}\cdot f(p,g(p)) \leq 2\uparrow\uparrow((p\cdot 2^{2p^2}\cdot 6+6)^p+10)$.

Finally, observe that when the algorithm terminates all vertices in graphs in $\mathcal{D}$ that do not belong to an $f(p,|X_t|)$-large equivalence class belong to $X_t$. Further, $S\subseteq X_t$. For the $f(p,|X_t|)$-large equivalence classes, $f(p,|X_t|)$ many $f(p,|X_t|)$-large groups are added to $O$. Hence the induced graph $G'=G[O]$ produced by the algorithm is a reduced subgraph as per Definition~\ref{def:red-graph}, as desired.
\end{proof}

The central lemma we will be proving in the next subsection is that every contraction sequence of a reduced graph $G'$ can be lifted to a contraction sequence of $G$ with the same width. We formalize this statement below.

\begin{lemma}\label{lem:extensionVI}
Let $G'$ be a reduced graph of $G$. Then $\tww(G)=\tww(G')$. Moreover, given a partition of $G'$ into $S \uplus S' \uplus L_1 \uplus \cdots \uplus L_{f(p,x)}$ and a contraction sequence $\C'$ of $G'$ with width $k$, we can compute in time $\bigoh((g(p)+p\cdot p \cdot 2^{2p^2})!\cdot |G|^2)$ a contraction sequence $\C$ of $G$ with width $k$.
\end{lemma}

Before proceeding towards the proof of Lemma~\ref{lem:extensionVI} (which will be our aim in the next two subsections), we show that establishing the lemma would allow us to prove Main Result~\ref{result2}.

\begin{theorem}
\label{thm:maintwo}
An optimal contraction sequence of an input graph $G$ can be computed in time $f(p)\cdot n$, where $p$ is the vertex integrity of $G$ and $f$ is a computable function.
\end{theorem}

\begin{proof}
Given a graph $G$, we first compute a deletion set $S$ using the known algorithm for computing vertex integrity~\cite{DrangeDH16}. Then, we invoke Lemma~\ref{lem:reducedcomputation} to compute a reduced graph $G'$. Recalling that $G'$ has order bounded by $2\uparrow\uparrow((p\cdot 2^{2p^2}\cdot 6+6)^p+10)$, we invoke Fact~\ref{fact:computecontract} to compute an optimal contraction sequence $\C'$ for $G'$ in time depending solely on $p$. Finally, we invoke Lemma~\ref{lem:extensionVI} to obtain a contraction sequence $\C$ for the original graph $G$ of the same width as $\C$. Since $\C'$ is an optimal contraction sequence for an induced subgraph of $G$, it follows that $\C$ must be an optimal contraction sequence of $G$, as desired.
\end{proof}

\subsection{Towards Lemma~\ref{lem:extensionVI}: The Ramsey Machinery}
\label{sub:lifting}
Recalling Lemma~\ref{lem:extensionVI}, let us consider a reduced graph $G'$ of the input graph $G$ such that $V(G')\neq V(G)$ (as otherwise the lemma is trivial). Hence, let $V(G') = S \uplus S' \uplus L_1 \uplus \cdots \uplus L_{f(p, |S\cup S'|)}$ be a partition witnessing that $G'$ is a reduced graph. Let $\mathcal{L}\coloneq \{L_1, \dots, L_{f(p, |S\cup S'|)}\}$. For brevity, we will hereinafter use \emph{large} as shorthand for $f(p, |S\cup S'|)$-large. Note that each of the vertex sets $L_i$, $i\in [f(p, |S\cup S'|)]$ forms a large group. 

For identification and tie-breaking purposes, it will be useful to fix an arbitrary total order $<^\#$ over $V(G)$. Among others, this immediately yields a total order of the equivalence classes $[\sim]$ of $\sim$ and a total order $\prec$ of the large groups in $L_1 \uplus \cdots \uplus L_{f(p, |S\cup S'|)}$ (both can be defined, e.g., by the order in which each vertex set is seen when following $<^\#$, however the specific way we induce these total orders does not matter, only their existence). It will also be useful to introduce a notion of ``canonical representative'' for each of the large equivalence classes in $[\sim]$; intuitively, one could imagine that these canonical representatives all occur in the same large group, but this is neither important nor necessary.

\begin{definition}
    Let $k$ be the number of large equivalence classes of $\sim$ and let $R_i$ 
    be an arbitrarily chosen but fixed canonical representative of the $i^{th}$ large  equivalence class of $\sim$. Let $R:=V(R_1)\uplus \cdots \uplus V(R_k)$.
\end{definition}

The set $R$ above will essentially be used as a sort of blank ``canvas'' for our future statements. In particular, its sole purpose is to have a natural isomorphism to $R$ that allows us to map consistently between different large groups and identify vertices between groups:

\begin{definition}
For a large group $X$ with $G[X]=H_1\uplus \cdots \uplus H_k$, $H_i\in [R_i]_\sim$ for each $i\in [k]$, let $\alpha_X$ be an isomorphism from $G[X]$ to $G[R]=R_1\uplus \cdots \uplus R_k$ such that for each $i\in [k]$ and vertex $u\in H_i$, $\alpha_X(u)\in V(R_i)$, and for each $v\in S$, $uv\in E(G)$ if and only if $\alpha_X(u)v\in E(G)$.    \end{definition}

For any large group $X$ and for each $u\in R$ we will sometimes use $u_X$ as shorthand for $\alpha_X^{-1}(u)$. 
From now, say $L' \prec L^\circ \prec L^\star$, in $\mathcal{L}$. The role of these will be to serve as an extended ``canvas'' for comparing triples of large groups.

\begin{definition}
For $X\prec Y\prec Z\in \mathcal{L}\setminus \{L',L^\circ,L^\star\}$, we define $\phi_{X,Y,Z}:S\cup S'\cup X\cup Y \cup Z \rightarrow S\cup S'\cup L'\cup L^\circ \cup L^\star$ as follows:
\begin{itemize}
        \item $\phi_{X,Y,Z}(s)=s$ for each $s\in S\cup S'$
        \item $\phi_{X,Y,Z}(x)=\alpha^{-1}_{L'}(\alpha_X(x))$ for each $x\in X$
        \item $\phi_{X,Y,Z}(y)=\alpha^{-1}_{L^\circ}(\alpha_Y(y))$ for each $y\in Y$
        \item $\phi_{X,Y,Z}(z)=\alpha^{-1}_{L^\star}(\alpha_Z(z))$ for each $z\in Z$        
\end{itemize}
Note that $\phi_{X,Y,Z}$ is an isomorphism from $G[S\cup S'\cup X\cup Y \cup Z]$ to $G[S\cup S'\cup L'\cup L^\circ \cup L'']$.
\end{definition}

Now, let us consider an arbitrary hypothetical ``solution'' to our problem on $G$---that is, a contraction sequence $\C'$ of $G'$ with minimum width. Given such $\C'$, we can characterize the behavior of a triple of large groups via a signature of size that is bounded solely by our parameter. We formalize this below (for a choice of $\C'$):

\begin{definition}
For $X\prec Y \prec Z \in \mathcal{L}\setminus\{L',L^\circ, L^\star\}$, 
  $\textsf{info}_{\C'}(X, Y,Z)$ is the contraction sequence on $G[S\cup S'\cup L' \cup L^\circ \cup L^\star]$ obtained from $\C'[S\cup S'\cup X \cup Y \cup Z]$ by applying $\phi_{X,Y,Z}$ on every vertex in every trigraph in $\C'[S\cup S'\cup X \cup Y \cup Z]$.
\end{definition}

Essentially, $\textsf{info}_{\C'}(X, Y,Z)$ is $\C'[S\cup S'\cup X \cup Y \cup Z]$ but translated into the shared canvas of $G[S\cup S'\cup X \cup Y \cup Z]$---as a consequence, for two distinct triples $X,Y,Z$ and $X',Y',Z'$ of large groups it may happen that $\textsf{info}_{\C'}(X, Y,Z)=\textsf{info}_{\C'}(X', Y',Z')$. 
Crucially, since the number of equivalence classes of $\sim$ is upper-bounded by $p \cdot 2^{2p^2}$ as per Fact~\ref{fact:size-equiv-class}, by Definition~\ref{def:red-graph} we have $|S\cup S'\cup X \cup Y \cup Z|\leq h(p,x)$ where $h(p,x)=x + 3p^2 \cdot 2^{2p^2}$. Hence, the number of possible  contraction sequences over $G[S\cup S'\cup X \cup Y \cup Z]$ can be upper-bounded by $\texttt{col}(p,x)=h(p,x)^{2\cdot h(p,x)}$.

The above considerations mean that for any set of three large groups, we can apply $\prec$ to obtain an ordering $X \prec Y \prec Z$ and then view $\textsf{info}_{\C'}(X, Y,Z)$ as a color from a set of at most 
$h(p, g(p))^{2\cdot h(p, g(p))}$ many colors, since $h(p, .)$ is increasing and $x\leq g(p)$.
We can then invoke Fact~\ref{fact:Ramsey} to show that the reduced graph contains a subset of $3p+3$ ``uniformly behaved'' large groups w.r.t.\ $\C'$, regardless of what $\C'$ is.

\begin{lemma}
\label{lem:nicelarge}
    Let $G'$ be a reduced graph of $G$ and $\C'$ be an arbitrary contraction sequence of $G'$ of width $w$. There exist distinct large groups $H_1,\dots,H_{3p+3} \in \mathcal{L} \setminus \{L',L^\circ,L^\star\}$ such that  for each $1\leq a<b<c\leq 3p+3$ and  each $1\leq a'<b'<c'\leq 3p+3$, $\textsf{info}_{\C'}(H_a, H_b, H_c)=\textsf{info}_{\C'}(H_{a'}, H_{b'}, H_{c'})$.    
\end{lemma}

\begin{proof}
Let $Q$ be an auxiliary complete hypergraph constructed as follows. The vertex set of $Q$ is the set of large groups in $\mathcal{L} \setminus \{L',L^\circ,L^\star\}$. For each set of three distinct large groups $\{X,Y,Z\}$ where $X\prec Y\prec Z$, create a hyperedge $\{X,Y,Z\}$ and label it with $\textsf{info}_{\C'}(X, Y,Z)$. Since $|V(Q)|=f(p,x)-3$ and in particular 

$$|V(Q)|\geq 2^{2^{2^{2^{2^{2^{xp}}}}}} \geq \texttt{col}(p,x)^{\texttt{col}(p,x)^{2\texttt{col}(p,x) \cdot (3p+3)}},$$ 
the claim follows from Fact~\ref{fact:Ramsey}.
 \end{proof}

We call a set of $3p+3$ large groups satisfying the conditions of Lemma~\ref{lem:nicelarge} \emph{uniform} (w.r.t.\ $\C'$). 
 Below, we show that while Lemma~\ref{lem:nicelarge} only stipulates that a uniform set satisfies a ternary form of uniformity, this in fact implies also uniformity of pairs and singletons from the set. Many of our arguments in the next subsection in fact only rely on these simpler forms of uniformity. Towards formalizing these, we define $\textsf{info}_{\C'}(X, Y)$ and $\textsf{info}_{\C'}(X)$ analogously as $\textsf{info}_{\C'}(X, Y, Z)$ but where the vertices are only mapped to $(L', L^\circ)$ and to $L'$, respectively.

\begin{lemma}
\label{lem:nicelargebis}
Let $H_1,\dots,H_{3p+3}$ be a uniform set of large groups. Then for each $1\leq i<j\leq 3p+3$ and each $1\leq k<\ell\leq 3p+3$, $\textsf{info}_{\C'}(H_i, H_j)=\textsf{info}_{\C'}(H_k, H_\ell)$ and for each $1\leq i\leq 3p+3$ and each $1\leq k \leq 3p+3$, $\textsf{info}_{\C'}(H_i)=\textsf{info}_{\C'}(H_\ell)$.
\end{lemma}

\begin{proof}
We first prove the claim for pairs of indices. For any $1 \leq a < b < c \leq 3p+3$, we have $\textsf{info}_{\C'}(H_a, H_b, H_c)= \textsf{info}_{\C'}(H_1, H_2, H_3)= \textsf{info}_{\C'}(H_1, H_3, H_4)=\textsf{info}_{\C'}(H_2, H_3, H_4)$ because of $H_1,\dots,H_{3p+3}$ being a uniform set of large groups as per Lemma~\ref{lem:nicelarge}. The two latter equalities imply that the contraction sequences $\C'[S\cup S' \cup H_1 \cup H_2]$, $\C'[S\cup S' \cup H_1 \cup H_3]$ and $\C'[S\cup S' \cup H_2 \cup H_3]$ are the equal up to renaming (i.e., up to transferring to $S\cup S' \cup L' \cup L^\circ$). At the same time, the first equality, when applied to $(a,b,c)=(i,j,x\notin \{i, j\})$ for $j< 3p+3$ implies that  $\C'[S\cup S' \cup H_i \cup H_j]$ is equal up to renaming to one of $\C'[S\cup S' \cup H_1 \cup H_2]$, $\C'[S\cup S' \cup H_1 \cup H_3]$ and $\C'[S\cup S' \cup H_2 \cup H_3]$, and by the previous argument, it is equal up to renaming to each of them. We then obtain easily that for each $1\leq i<j< 3p+3$ and each $1\leq k<\ell\leq 3p+3$, $\textsf{info}_{\C'}(H_i, H_j)=\textsf{info}_{\C'}(H_1, H_2)=\textsf{info}_{\C'}(H_k, H_\ell)$. The argument for the case where $j=3p+3$ is analogous whereas one begins by observing that
$\C'[S\cup S' \cup H_2 \cup H_3]$, $\C'[S\cup S' \cup H_2 \cup H_4]$ and $\C'[S\cup S' \cup H_3 \cup H_4]$ are equal up to renaming.

The reasoning is similar for the second claim: for any $1 \leq a < b \leq 3p+3$, $\textsf{info}_{\C'}(H_a, H_b)= \textsf{info}_{\C'}(H_1, H_2)= \textsf{info}_{\C'}(H_1, H_3)=\textsf{info}_{\C'}(H_2, H_3)$ because of the claim we just proved. We obtain that $\C'[S\cup S' \cup H_1]$ and $\C'[S\cup S' \cup H_2]$ are equal up to renaming (i.e., up to transferring to $S\cup S' \cup L'$). At the same time, the first equality, when applied to $\{a, b\}=\{i,j\neq i\}$ implies that  $\C'[S\cup S' \cup H_i]$ is equal up to renaming to $\C'[S\cup S' \cup H_1]$ or $\C'[S\cup S' \cup H_2]$, and by the previous argument, it is equal up to remaining to both of them. We then have that for each $1\leq i\leq 3p+3$ and each $1\leq k\leq 3p+3$, $\textsf{info}_{\C'}(H_i)=\textsf{info}_{\C'}(H_1)=\textsf{info}_{\C'}(H_k)$.
\end{proof}

\subsection{The Proof of Lemma~\ref{lem:extensionVI}}
 
Let $\mathcal{F} \subseteq \mathcal{L} \setminus \{L', L^\circ,L^\star\} $ be the family of distinct large groups whose existence is given by \cref{lem:nicelarge} with $\mathcal{F}\coloneq (B_1, \dots, B_\omega)$, where $B_i \prec B_j$ if $i<j$.
Let $\C^*$ be the restriction of $\C'$ to $V(\bigcup_i B_i \cup S \cup S')$.

\begin{definition}\label{def:types}
We say that a contraction $c$ in $\C^*$ is:\begin{itemize}
\item \emph{unique} if it contracts bags intersecting $S\cup S'$ together; otherwise
\item \emph{special} if it contracts a bag intersecting a given $B_i$ with a bag intersecting $S \cup S'$; otherwise
\item \emph{vertical} if it contracts two bags intersecting a given $B_i$; otherwise
\item \emph{horizontal} if it contracts a bag $x$ with a bag $y$ such that $\exists v_0\in \beta(x)$ and $\exists w_0\in \beta(y)$ such that $\alpha(v_0)=w_0$; and
 \item \emph{diagonal} if none of the above cases apply.  \end{itemize}
\end{definition}

Intuitively, \emph{unique} contractions will form a small number of ``milestones'' in our considered contraction sequence, \emph{special} contractions merge vertices from large groups with $S\cup S'$, \emph{vertical} contractions merge vertices from the same large group and \emph{horizontal} contractions merge corresponding vertices across large groups. We now prove several useful properties of such contractions, including the fact that diagonal contractions can be disregarded; however, before that it will be useful to introduce some additional terminology.
For the following, it will be useful to recall that for a vertex $u\in R$ and large group $X$, we use the notation $u_X$ to denote the counterpart of $u$ in $X$.

\begin{definition}
             We say that a contraction $c$ of the bags $b_1$ and $b_2$ \emph{contracts an unordered pair} $\{x, y\}$ if $x\in \beta(b_1)$ and $y\in \beta(b_2)$. Moreover $c$ \emph{involves a \blob $B_i$} if it contracts a pair intersecting $B_i$.
             
Given an unordered pair $\{a,b\}\in {V(G) \choose 2}$ of vertices, we introduce the set $\textsf{Corr}_{\{a,b\}}$ of \emph{corresponding pairs} as follows:
    \begin{itemize}
\item for $\{a,b\}=\{x_{B_i},s\}$ for some $x\in R$ and $s \in S\cup S'$: $\textsf{Corr}_{\{a,b\}} = \{\{x_{B_1}, s\},\{x_{B_\omega}, s\}\}$, 
\item for $\{a,b\}=\{x_{B_i}, y_{B_j}\}$ for some $x,y\in R$ where $i<j$: $\textsf{Corr}_{\{a,b\}}= \{\{x_{B_1}, y_{B_2}\},\{x_{B_{\omega-1}},y_{B_\omega}\}\}$, and 
\item for $\{a,b\}=\{x_{B_i}, y_{B_i}\}$ for some $x,y\in R$: $\textsf{Corr}_{\{a,b\}} = \{\{x_{B_1}, y_{B_1}\}, \{x_{B_\omega}, y_{B_\omega}\}\}$.
\end{itemize}
\end{definition}

We note that, while the $\textsf{info}$ notation used in Lemma~\ref{lem:nicelargebis} establishes a notion of similarity between the behavior of large groups, $\textsf{Corr}$ allows us to use the first and last large group as markers in order to facilitate our further proof arguments. Indeed, we will show that $B_1$ and $B_\omega$ have a special role: 
each contraction outside of ${S\cup S' \choose 2}$  has a preceding counterpart which involves et least one of $B_1$ and $B_\omega$.
The proofs of the next two lemmas rely heavily on the Ramsey machinery of Subsection~\ref{sub:lifting}. In particular, the reason one needs to consider Ramsey hypergraph theory is to obtain Lemma~\ref{lem:nodiag}.

\begin{lemma}
\label{lem:first-occurence}
 Let $c$ be any non-unique contraction and $\{x,y\}$ be a pair contracted by $c$. There is $\{x',y'\} \in \textsf{Corr}_{\{x,y\}}$ that is either contracted by $c$ or had already been contracted by an earlier contraction in the sequence. 
\end{lemma}

\begin{proof}
First, let us suppose that $c$ is special and contracts $\{x_{B_i}, s\}$, for $1<i<\omega$. When looking at the pairs $(B_1, B_i)$ and $(B_i, B_\omega)$, it is clear that to avoid any contradiction with \cref{lem:nicelargebis}, either $\{x_{B_1}, s\}$ or $\{x_{B_\omega}, s\}$ has to be contracted before (or at the same time as, meaning that they occur in the same bags) $\{x_{B_i}, s\}$. If $c$ is vertical and contracts $\{x_{B_i}, y_{B_i}\}$, the argument translates immediately.

Consider the case where $c$ is horizontal or diagonal and contracts $\{x_{B_i}, y_{B_j}\}$. If $1<i<j<\omega$, when looking at $(B_1, B_i, B_j)$ and $(B_i, B_j, B_\omega)$, it is clear that $\{x_{B_1}, y_{B_i}\}$ or $\{x_{B_j}, y_{B_\omega}\}$ has to be contracted before or at the same time as $\{x_{B_i}, y_{B_j}\}$. The other case being symmetric, let us assume that we are in the first case $\{x_{B_1}, y_{B_i}\}$ being contracted before $\{x_{B_i}, y_{B_j}\}$. If $i=2$, then the lemma follows immediately by the first sentence. Otherwise we look at $(B_1, B_2, B_i)$ and $(B_1, B_i, B_j)$. The pair $\{x_{B_i}, y_{B_j}\}$ is contracted after (or simultaneously with) $\{x_{B_1}, y_{B_i}\}$ in the second triple, implying in the first triple that $\{x_{B_2}, y_{B_i}\}$ is contracted after (or simultaneously with) $\{x_{B_1}, y_{B_2}\}$. However, in the triple $(B_2, B_i, B_j)$ the same implies that $\{x_{B_2}, x_{B_i}\}$ is contracted before (or simultaneously with) $\{x_{B_i}, x_{B_j}\}$, and by chaining these precedences we conclude the proof.

If $1=i<j=\omega$, let us suppose that $\{x_{B_1}, y_{B_\omega}\}$ is contracted and neither $\{x_{B_1}, y_{B_2}\}$ nor $\{x_{B_{\omega-1}}, y_{B_\omega}\}$ are. This means that in $(B_1, B_2, B_{\omega-1})$ (resp. in $(B_2, B_{\omega-1}, B_\omega)$), $\{x_{B_1}, y_{B_{\omega-1}}\}$ (resp. $\{x_{B_2}, y_{B_\omega}\}$) has to be contracted strictly before both $\{x_{B_1}, y_{B_2}\}$ and $\{x_{B_2}, y_{B_{\omega-1}}\}$ (resp. both $\{x_{B_2}, y_{B_{\omega-1}}\}$ and $\{x_{B_{\omega -1}}, y_{B_\omega}\}$). In particular, both $\{x_{B_1}, y_{B_{\omega-1}}\}$ and $\{x_{B_2}, y_{B_\omega}\}$ are contracted strictly before $\{x_{B_2}, y_{B_{\omega-1}}\}$. However, this is not possible, as after these contractions, $x_{B_2}$ has been contracted with $y_{B_\omega}$, which has itself been contracted with $x_{B_1}$, and the chain continues with $y_{B_{\omega-1}}$. Thus, we obtain a contradiction, as $\{x_{B_2}, y_{B_{\omega-1}}\}$ turns out to already be contracted.

If $1=i<j<\omega$, we consider $(B_1, B_2, B_j)$ and $(B_1, B_j, B_\omega)$ to observe that $\{x_{B_1}, y_{B_2}\}$ or $\{x_{B_j}, y_{B_\omega}\}$ has to be contracted before or at the same time as $\{x_{B_1}, y_{B_j}\}$. In the first case we already have what we want, in the other we iterate our observation on this new contraction in a symmetrical way, to obtain that $\{x_{B_1}, y_{B_j}\}$ or $\{x_{B_{\omega-1}}, y_{B_\omega}\}$ has to be contracted before or at the same time as $\{x_{B_j}, y_{B_\omega}\}$. Chaining these precedences, we obtain either that $\{x_{B_1}, y_{B_j}\}$ happens after or at the same time as $\{x_{B_{\omega-1}}, y_{B_\omega}\}$ or $\{x_{B_1}, y_{B_\omega}\}$, and the latter can only happen after either $\{x_{B_1}, y_{B_2}\}$ or $\{x_{B_{\omega-1}}, y_{B_\omega}\}$ according to the previous paragraph. This concludes the proof of this case, and the case $1<i<j=\omega$ is symmetric.
\end{proof}

\begin{lemma}
\label{lem:nodiag}
No diagonal contraction can be present in $\C^*$.
\end{lemma}

\begin{proof}
Let us consider the first diagonal contraction $c$ contracting a pair $\{x_i, y_j\}$ with $i<j$. 
By \Cref{lem:first-occurence}, we know that $(i,j)$ is either $(1,2)$ or $(\omega-1, \omega)$. Let us consider the first case, the other being symmetric.
In $(B_1, B_2, B_\omega)$, $\{x_1, y_2)$ is contracted before (or simultaneously with) $\{x_1, y_\omega\}$ and $\{x_2, y_\omega\}$. Observe that $\{x_1, y_\omega\}$ and $\{y_2, y_\omega\}$ are contracted by the same contraction $c'$ since $x_1$ and $y_2$ will already be contracted together. Let us now assume that $c'$ happens before (resp. is) the contraction $c''$ contracting $\{x_2, y_\omega\}$. This implies, because $(B_1, B_2)$ and $(B_2, B_\omega)$ follow the same pattern, that $\{y_1,y_2\}$ was contracted before (resp. at the same time as) $c$, which makes $c$ vertical (resp. horizontal). Second case, $\{x_2, y_\omega\}$ is contracted by $c''$ before $\{x_1, y_\omega\}$ is contracted by $c'$. Then $\{x_1, y_\omega\}$ and $\{x_2, y_2\}$ are contracted by the same contraction $c'$ since $x_2$ and $y_\omega$ are already contracted together. By \cref{lem:first-occurence}, $\{x_1,y_1\}$ or $\{x_\omega, y_\omega\}$ was contracted before (resp. by) $c'$. This implies that $c'$ contracts $\{y_1, y_\omega\}$ or $\{x_1, x_\omega\}$, and since $(B_1, B_\omega)$ and $(B_1, B_2)$ follow the same pattern, it means that $c$ also contracted $\{x_1,x_2\}$ or $\{y_1, y_2\}$, which is a contradiction with $c$ being diagonal.
\end{proof}

The next definition allows us to identify certain milestones that in turn leads to the crucial notion of \emph{blocks} of $\C^*$.
\begin{definition}\label{def:blocks}
We say that two pairs $(A,B\neq A)$ and $(C,D\neq C)$ of large groups are \emph{synchronized} at a trigraph $G_s$ in $\C^*$ if the trigraphs induced on $V(S \cup S' \cup A \cup B)$ and $V(S \cup S' \cup C \cup D)$  are identical up to renaming via $\alpha$. 
 
Let $H_1,\dots,H_{3p+3}$ be a uniform set of large groups w.r.t.\ the contraction sequence $\C^*$. We say that a segment of $\C^*$ is a \emph{block} if it is a minimal segment of size at least $2$---i.e., one that contains at least one contraction---such that in its first and last trigraphs, the following holds:
Every two pairs $(H_i, H_j)$, $(H_k, H_\ell)$ for $i,j,k,\ell\in [3p+3]$, $i\neq j$, $k\neq \ell$ are synchronized.
 \end{definition}

For brevity, we say that two large groups $A\neq C$ are \emph{synchronized} at a trigraph $G_s$ in $\C^*$ if the trigraphs induced on $V(S \cup S' \cup A)$ and $V(S \cup S' \cup C)$  are identical up to renaming via $\alpha$. We note that this notion of synchronization of singletons also occurs at the beginning and end of blocks as it is directly implied by the synchronization of pairs.

It is easy to see that, by definition, all pairs are synchronized on $G^*$ and $K_1$, thus $\C^*$ contains at least one block. Moreover, blocks form a ``near-partitioning'' of $\C^*$ where each pair of blocks can only intersect in their first or last trigraphs. The next definition provides a way of identifying pairs of vertices which are ``equivalent'' to each other given our decomposition into large groups; we remark that the permutation is merely used to avoid undesirable choices of the target indices, and that the defined equivalence relation is in fact the transitive, symmetric and reflexive closure of the relation induced by \textsf{Corr}.

\begin{definition}
\label{def:equiv_pairs}
Let $\{x,y\}$ be a pair of vertices. We call a pair $\{x',y'\}$ equivalent to $\{x,y\}$ if there exists a permutation $\sigma$ of $[\omega]$ such that: $\left(\exists u\in R,  \exists i \in [\omega], x=u_{B_i} \right) \Rightarrow x'=u_{B_{\sigma(i)}}$ and $\left(\exists v \in R, \exists j \in [\omega], y=v_{B_j} \right) \Rightarrow y'=v_{B_{\sigma(j)}}$. 
\end{definition}

\begin{observation}\label{lem:block_general}
If a block contains a contraction of a pair $\{x,y\}$,
it also contains contractions of all the equivalent pairs as per \cref{def:equiv_pairs}. Moreover, if a segment starting in a synchronized state is a minimal segment containing contraction of all the equivalent pairs of each pair contracted within it, this segment is a block.
\end{observation}

\begin{proof}
The first direction is easy to check: if a contracted pair involving $B_i$ (resp.\ $B_i$ and $B_j$) does not have an equivalent pair in $B_\ell$ (resp.\ $B_\ell$ and $B_k$), then $B_i$ and $B_\ell$ (resp. $(B_i, B_j)$ and $(B_\ell, B_k)$) are not synchronized. 
In the other direction, let us suppose that the segment satisfies the property but is not a block: by minimality of the segment, and the first direction that we already proved, there can be no block included in the segment. Thus, it means that the reason it is not a block is that there exists a pair $(B_i, B_j)$, $(B_\ell, B_k)$ of pairs of \blobs which are not synchronized at the end of the segment. Now, let us take a look at the difference between the two induced trigraphs on which the definition rely. Since these two trigraphs started synchronized, and for each pair being contracted in one the equivalent one in the other was also contracted, we obtain a contradiction. 
\end{proof}

\begin{observation}\label{lem:block_unique}
If a block contains a contraction which is unique, it contains only this contraction.
\end{observation}

\begin{proof}
Observe that, for any pair contracted by a unique contraction $c$ directly after a synchronized state, all of its equivalent pairs are contracted by $c$. Indeed, otherwise there exists a pair $\{x,y\}$ contracted by $c$ and an equivalent pair $\{x',y'\}$ not contracted by $c$. However, since this $c$ is a unique contraction, $x$ (resp. $y$) is part of the same bag as some $s$ (resp. $t$) belonging to $S\cup S'$. Because of the synchronized state before $c$, this means that $x'$ (resp. $y'$) also belongs to the same bag as $s$ (resp. $t$), and thus $\{x',y'\}$ is contracted by $c$. Thus, a block starting with a unique contraction contains no other, as a direct consequence of \cref{lem:block_general} and minimality of blocks.\\
If a block contains at least one unique contraction $c$, and a non-empty set of other contractions, there has to exist a pair $\{x,y\}$ and a contraction $c'$ of $\{x,y\}$ happening before $c$ such that the contraction $c'_e$ of one pair equivalent to $\{x,y\}$ happens after $c$, otherwise there would be a strictly smaller block. However, in the \blobs that they respectively affect, $c'$ and $c'_e$ would then happen in different relative orders with respect to $c$, which is a contradiction.
\end{proof}

\begin{lemma}\label{lem:block_v_monotone}
If a block contains vertical (resp. special) contractions, all such contractions will appear in order $C_{\sigma(1)}, C_{\sigma(2)}, \dots, C_{\sigma(\omega -1)}, C_{\sigma(\omega)}$ where: 
\begin{itemize}
\item each contraction in $C_{\sigma(i)}$ is a vertical contraction of a pair of vertices of $B_{\sigma(i)}$ (resp. a special contraction involving $B_(i)$), and
\item $C_i$ and $C_j$ are equivalent w.r.t. $\alpha$, and 
\item $\sigma(\circ)$ is either the identity $id(\circ)$ or $\omega + 1 -id(\circ)$, \ie, $\sigma(\circ)$ either follows or reverses the order defined by $\prec$.
\end{itemize}
\end{lemma}

\begin{proof}
We write the proof for the vertical case. The one for the special case is exactly the same, only needing to exchange the occurrences of vertical and special.\\
First, because of \cref{lem:block_general}, we know that the set of pairs contracted by the vertical contractions happening for each $B_i$ is the same. By choice of the family of \blobs, we also know that, among the contractions of pairs associated to a given \blob, the order is always the same, giving us that the sequences $C_i$ of the vertical contractions involving $B_i$ are indeed $\alpha$-equivalent. Still to prove is that all of these $C_i$ are happening is the right order, namely either with increasing or decreasing index. Let us suppose that there are $i, j, k$ such that the equivalent pairs (contracted by $c_i$, $c_j$ and $c_k$ respectively) happen in a non-monotone order. This is a direct contradiction to $(B_i, B_j)$, $(B_i, B_k)$ and $(B_j, B_k)$ following the same pattern since the pair would then for some of them happen first on the left side, and for others on the right side. Let us suppose the pairs contracted in the first vertical contraction $c$ happening in the block are in $B_1$ (the other case is symmetric, see \cref{lem:first-occurence}). Then, let us consider the contractions happening before $c_2$. Each of these must contract a pair from $B_1$, since the first vertical contraction involving any given $B_i$ will need to be $c_i$ to follow the same pattern as $B_1$, and none of these $c_{i\neq 1}$ has happened yet. 

Next, between $c_2$ and $c_3$ the vertical contractions are contracting precisely the equivalent pairs of $B_2$, and so on, until the contraction $c_\omega$ which is followed by the vertical contractions of the equivalent pairs of $B_\omega$. The remaining point to prove is that all of the vertical contractions that we considered are already accounted for, and that there doesn't exist a set of equivalent contractions happening after. To this end, we will prove that there is an end of a block before the next vertical contraction. Indeed, this happens if there is a trigraph where all pairs of \blobs are synchronized, and if this is not the case, there exist a contraction $c'$ (not vertical) of a pair such that its equivalent pairs are only ``partially completed'' at the end of this first wave of vertical contractions. By pattern uniformity, all the contractions of the equivalent pairs have already been executed outside of those involving $B_\omega$ (note that if $c'$ is horizontal, it misses at most the contraction of all $x_{i<\omega}$ with $x_\omega$, which in particular is the contraction $\{x_1, x_\omega)$). Again by pattern uniformity, no new contraction (\ie, with no equivalent contraction which happened previously) $c^*$ can happen before this contraction involving $B_\omega$ happens, since it would then disrupt pattern uniformity between one of the following pairs of tuples: \begin{itemize}
\item $(B_1, B_2, B_3)$ and $(B_i, B_j, B_\omega)$ if $c^*$ involves $B_i$ and $B_j$ and $i<j<\omega$, or
\item $(B_1, B_2)$ and $(B_i, B_\omega)$ if $c^*$ involves $B_i$ and $i<\omega$, or
\item $(B_1, B_2)$ and $(B_i, B_\omega)$ if $c^*$ involves $B_i$ and $B_\omega$, or
\item $(B_1)$ and $(B_\omega)$ if $c^*$ involves $B_\omega$.
\end{itemize}
With these observations, we conclude that before any new vertical contraction can happen, the contraction sequence must first complete the application of all classes of equivalent contractions which are already started---and once those are completed, \cref{lem:block_general} guarantees that the block ends.
\end{proof}

\begin{lemma}\label{lem:block_h_monotone}
If a block contains horizontal contractions, all such contractions will appear in order \\$C_{\sigma(1),\sigma(2)}, C_{\sigma(2),\sigma(3)}, \dots, C_{\sigma(\omega-2),\sigma(\omega -1)}, C_{\sigma(\omega-1),\sigma(\omega)}$ where 
\begin{itemize}
\item each contraction in $C_{\sigma(i),\sigma(i+1)}$ is horizontal and involves $B_{\sigma(i)}$ and $B_{\sigma(i+1)}$, 
\item $C_{i,i+1}$ and $C_{j,j+1}$ are equivalent w.r.t. $\alpha$, and 
\item $\sigma(\circ)$ is either the identity $id(\circ)$ or $\omega + 1 -id(\circ)$ (\ie, either following or reversing the order defined by $\prec$).
\end{itemize}
\end{lemma}

\begin{proof}
The proof follows by adapting the argument from the proof of \cref{lem:block_v_monotone}.\\
First, because of \cref{lem:block_general}, we know that the set of pairs contracted by horizontal contractions for each pair of $B_i, B_j$ is the same; we only consider consecutive choices of $i,j$ in the lemma statement because of the behavior that emerges in the proof. By choice of the family of \blobs, we also know that, among the contractions involving a given pair of \blobs, the order is always the same, giving us that the sequences $C_{i,j}$ of horizontal contractions involving $B_i$ and $B_j$ are indeed $\alpha$-equivalent. It remains to prove that all of the $C_{i,j}$ are happening in the right order, namely either with increasing or decreasing index, and that we care only about consecutive \blobs. Let us suppose that there exists an $i$ such that the contractions of the pairs $\{x_{B_i}, x_{B_{i+1}}\}$, $\{x_{B_{i+1}}, x_{B_{i+2}}\}$ and $\{x_{B_{i+2}}, x_{B_{i+3}}\}$ happen in a non-monotone order. This is a direct contradiction to $(B_i, B_{i+1}, B_{i+2})$ and $(B_{i+1}, B_{i+2}, B_{i+3})$ following the same pattern since the pairs of $x$ would then for one of them be contracted on the left side first, and for the other on the right side first. This local monotonicity implies here that the way horizontal contraction work is here (with the example of increasing index) $\{x_{B_1}, x_{B_2}\}$, then $\{x_{B_1}=x_{B_2}, x_{B_3}\}$, and $\{x_{B_1}=x_{B_2}=x_{B_3}, x_{B_4}\}$ and so on. 
Due to this, it is irrelevant when the pair $\{x_{B_i}, x_{B_j}\}$ is contracted for $i+1< j$ as this simply happens simultaneously with the same contraction as the pair $\{x_{B_{j-1}}, x_{B_j}\}$. 

Let us now suppose the first horizontal contraction $c$ that happens in the block involves $B_1$ and $B_2$ (the other case is symmetric, see \cref{lem:first-occurence}). Let, for $i\in [\omega]$ $c_(i)(i+1)$ be the contraction of the equivalent pair between the \blobs $B_i$ and $B_{i+1}$. Then, let us consider the pairs contracted by horizontal contractions before $c_{(2)(3)}$. All of these must be pairs of a vertex in $B_1$ and a vertex in $B_2$, since the first horizontal contraction between $(B_i, B_j)$ needs to be $c_{(i)(j)}$ to follow the same pattern as $(B_1, B_2)$, and none of these have happened yet. Then, between $c_{(2)(3)}$ and $c_{(3)(4)}$ the pairs contracted by horizontal contractions are precisely the same as between $B_2$ and $B_3$, and so on, until the contraction $c_{(\omega-1)(\omega)}$ which is followed by the contractions of the equivalent pairs involving $B_{\omega-1}$ and $B_\omega$. 

The remaining point to prove is that all contractions of the equivalent pairs involving $B_{\omega-1}$ and $B_\omega$ happen before any ``new'' horizontal contraction. Namely, we will prove that the block ends before the next horizontal contraction takes place. Indeed, this happens if there is a trigraph where all pairs of \blobs are synchronized, and if this is not the case, there exists a contracted pair (not horizontal) such that its equivalent pairs are only partially contracted at the end of this first wave of horizontal contractions. By pattern uniformity, all equivalent pairs have already been contracted except for those involving $B_\omega$. Again by pattern uniformity, no new pair (\ie, with no equivalent pair previously contracted) can be contracted by a contraction $c^*$ before the pair involving $B_\omega$ is, since it would then contradict the pattern uniformity between: \begin{itemize}
\item $(B_1, B_2, B_3)$ and $(B_i, B_j, B_\omega)$ if $c^*$ involves $B_i$ and $B_j$ and $i<j<\omega$, or
\item $(B_1, B_2)$ and $(B_i, B_\omega)$ if $c^*$ involves $B_i$ and $i<\omega$, or
\item $(B_1, B_2)$ and $(B_i, B_\omega)$ if $c^*$ involves $B_i$ and $B_\omega$, or
\item $(B_1)$ and $(B_\omega)$ if $c^*$ involves $B_\omega$.
\end{itemize}
With the above observations, we conclude that before any new horizontal contraction can happen, we first need to complete the contraction of all classes of equivalent pairs which have already been partly contracted. Moreover, once we complete these, \cref{lem:block_general} guarantees that the current block comes to an end, thus concluding the proof.
\end{proof}

\begin{lemma}
\label{lem:structblock}
A block containing no unique contraction is of the form: \\$C_{\sigma(1)},C_{\sigma(1),\sigma(2)}, C_{\sigma(2)},  C_{\sigma(2),\sigma(3)}, \dots, \dots, C_{\sigma(\omega -1)}, C_{\sigma(\omega-1),\sigma(\omega)}, C_{\sigma(\omega)}$ \\
where for each $i,j\in [\omega]$:
\begin{itemize}
\item each $C_{\sigma(i)}$ contains only vertical and special contractions involving $B_{\sigma(i)}$, 
\item $C_i$ and $C_j$ are equivalent w.r.t. $\alpha$, 
\item $C_{\sigma(i),\sigma(i+1)}$ contains exclusively horizontal contractions between $B_{\sigma(i)}$ and $B_{\sigma(i+1)}$, 
\item $C_{i,i+1}$ and $C_{j,j+1}$ are equivalent w.r.t. $\alpha$, and 
\item $\sigma$ is either the identity or $\omega + 1$ minus the identity (reversing the order defined by $\prec$).
\end{itemize}
\end{lemma}

\begin{proof}
First, combining \cref{lem:block_v_monotone} in its vertical and special versions, it is not hard to see that to fit within one block, the order of affected \blobs for vertical and special contractions has to be the same. We thus let these contractions mix together (as long as they are about the same \blob) and this defines our $C_i$ for each $i$. Note that, similarly as in the preceding proofs of this section, all $C_i$ are equivalent \emph{w.r.t.}~$\alpha$. Now, when we additionally consider \cref{lem:block_h_monotone}, the forced order is also clear, and how each $C_{i,i+1}$ fills the gap between $C_i$ and $C_{i+1}$ can also easily be observed: if a horizontal contraction of some pair $\{x_1, x_2\}$ happens before a contraction of $C_1$, we would need for pattern uniformity (between $(B_1, B_2)$ and $(B_1, B_\omega)$) that $(x_{B_1}, x_{B_\omega}\}$ also happens before the same contraction of $C_1$, and then this should happen before the entirety of $C_1$ since it happens before the entirety of $C_2$. If this is the case, gathering all such horizontal contractions, we would obtain that they form a separate block at the beginning of the currently considered block, thus contradicting the minimality of the considered block. Analogously, one also reaches a contradiction if there were to exist a horizontal contraction of the pair $\{x_{B_1}, x_{B_2}\}$ after a contraction in $C_2$, this time due to a block of horizontal contractions forming at the end of the currently considered block. 
\end{proof}

As an intermediate step towards proving Lemma~\ref{lem:insertblob}---which allows us to lift ``well-behaved'' contraction sequences towards $G$ (or a supergraph thereof), we prove the following lemma.
\begin{lemma}\label{lem:insertblock}
Given $k\in \mathbb{N}$, and $\seg$ a block in $\C^*$ and the trigraphs $G^+_s$ (resp. $G^+_e$) corresponding to $\seg_s$ (resp. $\seg_e$) the first (resp. the last) trigraph of $\seg$, in which we add $k$ \blobs in the exact same state as the \blobs already present. We can construct in linear time a partial contraction sequence $\C^+$ from $G^+_s$ to $G^+_e$.  \end{lemma}
\begin{proof}
Using \cref{lem:structblock}, we know that the contractions happening in $\seg$ are either of the form
$C_{\sigma(1)},C_{\sigma(1),\sigma(2)}, C_{\sigma(2)},  C_{\sigma(2),\sigma(3)}, \dots, \dots, C_{\sigma(\omega -1)}, C_{\sigma(\omega-1),\sigma(\omega)}, C_{\sigma(\omega)}$, or consist of a single unique contraction as per Lemma~\ref{lem:block_unique}. In the former case, we extend it to $G^+_s$ in this way: $C_{\sigma'(1)},C_{\sigma'(1),\sigma'(2)}, C_{\sigma'(2)},  C_{\sigma'(2),\sigma'(3)}, \dots, \dots, C_{\sigma'(k+\omega -1)}, C_{\sigma'(k+\omega-1),\sigma'(k+\omega)}, C_{\sigma'(k+\omega)}$, where $\sigma'$ is the identity if $\sigma$ is the identity, and otherwise $\sigma'(x)=k+\omega+1-x$. It is then fairly evident that this yields at the end the trigraph $G^+_s$, and the same also holds for the latter case (where no modifications are required). 
\end{proof}

\begin{lemma}
\label{lem:insertblob}
Given $k\in \mathbb{N}$, $\C^*$ a contraction sequence of width $t$ for $G^*$ such that all large groups in $G^*$ are uniform for $\C^*$, and $G_\emph{super}$ the graph obtained from $G^*$ by adding $k$ \blobs. We can construct in linear time a contraction sequence $\C_\emph{super}$ of width $t$ for $G_\emph{super}$.
\end{lemma}

\begin{proof}
We use \cref{lem:insertblock} for each block of $\C^*$ to obtain $\C_\emph{super}$ in linear time. 
Now we want to prove that the red degree in $\C_\emph{super}$ never exceeds $t$. Towards this, assume for a contradiction that there exists a trigraph $G_i$ where a vertex $q$ has red degree at least $t+1$. Let $N$ be a set of red neighbors of $q$ of size exactly $t+1$. For $n\in N$, we know by definition of twin-width that there is either (1) $q_x\in \beta(q)$, $n_0, n_1 \in \beta(n)$ such that precisely one of the latter is non-adjacent to $q_x$ in $G_\emph{super}$, or (2) $q_0, q_1\in \beta(q)$, $n_x \in \beta(n)$ such that precisely one of the former is non-adjacent to $n_x$ in $G_\emph{super}$. Let us select a triple $\{i,j,k\}$ such that $q_x, n_0, n_1$ (resp.\ $q_0, q_1, n_x$) belong to $V(S \cup S' \cup B_i \cup B_j \cup B_k)$; if multiple choices for $\{i,j,k\}$ exist we choose one arbitrarily. Let $\mathcal{I}^\circ$ be the union of all such selected triples over all $n\in N$---in particular, $\mathcal{I}^\circ$ contains at most $3t+3$ indices of large groups. Let $\mathcal{I}$ be obtained by adding to $\mathcal{I}^\circ$ arbitrary indices to obtain a set of size precisely $3t+3$.

The second step of our proof is to observe that in the restriction $\C_\emph{restr}$ of $\C_\emph{super}$ to $V(S \cup S' \cup (B_i)_{i\in \mathcal{I}})$, the corresponding trigraph also contains a vertex with red degree at least $t+1$, since all of the red edges are also present there by our choice of $\mathcal{I}$. Moreover, using the specific shape of each block as established in \cref{lem:structblock}, we see that $\C_\emph{restr}$ and $\C^*$ are the same, up to a simple renaming: each element of $\mathcal{I}$ is replaced by its position within $i$, \ie the smallest $i\in \mathcal{I}$ is replaced by $1$, the second smallest by $2$ and so on.

The final step is to remark that the renaming described in the previous paragraph does not affect red edges in any way, implying that $\C^*$ contains a trigraph with red degree at least $t+1$---thus achieving the desired contradiction.
\end{proof}

We can now prove \cref{lem:extensionVI}, which was the last missing piece required for \cref{thm:maintwo}.
\begin{proof}[Proof of \cref{lem:extensionVI}]
Let $G'$ be a reduced graph of $G$, a partition $S \uplus S' \uplus L_1 \uplus \cdots \uplus L_{f(p,x)}$ of its vertices.   Using \cref{lem:nicelarge}, we know that for any hypothetical solution of width $k$ for $G'$, there exists a uniform set of \blobs of size $3p+3$. 
Thus, we can define $G_\emph{core}$ as the subgraph of $G'$ (and of $G$) containing only $S$, $S'$ and an arbitrary set of $3p+3$ \blobs, and brute-force a contraction sequence $\C_\emph{core}$ of width $k$ for $G_\emph{core}$, such that the group of \blobs in $G_\emph{core}$ is uniform for $\C_\emph{core}$. This brute-force computation takes time at most $\bigoh((g(p)+p\cdot p \cdot 2^{2p^2})!)$.
      Let us set $m$ the maximum size of a $\sim$ equivalence class in $G$. We define $G_\emph{super}$ the supergraph of $G$ such that all equivalence classes of $\sim$ which are not in $S'$ are of size $m$. We use \cref{lem:insertblob} to create a contraction sequence $\C_\emph{super}$ for $G_\emph{super}$, since it can be obtained from $G_\emph{core}$ by adding $m-k^*$ \blobs, and $\C_\emph{core}$ has the structure argued in Lemmas~\ref{lem:block_unique}-\ref{lem:structblock}. 
 This takes at most linear time in $|\C_\emph{super}|$ which is $\bigoh(|G|^2)$, and we can in the same amount of time take a restriction of $\C_\emph{super}$ to $G$, obtaining the sought-after contraction sequence of width $k$ for $G$.
\end{proof}

\paragraph*{Acknowledgments}
The authors want to thank Malory Marin for bringing oriented twin-width to the authors' attention, as well as Édouard Bonnet for his help in understanding the (computational) connection to the classical twin-width.

\bibliographystyle{plainurl}
\bibliography{oriented_arxiv_main}

\end{document}